%% file: main.tex
\documentclass[a4paper, onecolumn, colorlinks]{preprint}

\usepackage{mymacros,mytikzconfig}
\usepackage{graphicx}

\usepackage{enumitem}
\usepackage{comment}
\usepackage{microtype}

\usepackage[printonlyused]{acronym}

\makeatletter
\@ifundefined{AC@verridelabel}{}{%
  \let\org@overridelabel\AC@verridelabel
  \renewcommand*{\AC@verridelabel}[1]{%
    \@bsphack
    \protected@write\@auxout{}{\string\AC@undonewlabel{#1@cref}}%
    \org@overridelabel{#1}%
    \@esphack
  }%
}
\makeatother

\usepackage{booktabs}
\usepackage{arydshln}
\usepackage{multirow}
\usepackage{subcaption}

\usepackage{amsthm}
\usepackage{mathdots}
\usepackage{mathtools}

\allowdisplaybreaks

\AddToHook{begindocument/end}{
\theoremstyle{plain}
\newtheorem{theorem}{Theorem}[section]
\newtheorem{lemma}{Lemma}[section]

\newtheorem{proposition}{Proposition}[section]
\theoremstyle{definition}
\newtheorem{definition}{Definition}[section]
\newtheorem{assumption}{Assumption}[section]

\newtheorem{remark}{Remark}[section]

\crefformat{assumption}{#2Assumption~#1#3}
\Crefformat{assumption}{#2Assumption~#1#3}
\crefname{appendix}{Appendix}{Appendices}
\Crefname{appendix}{Appendix}{Appendices}
}

\makeatletter
\AddToHook{cmd/appendix/before}{
\def\cref@section@alias{appendix}\def\cref@subsection@alias{appendix}} 
\makeatother

\makeatletter
\let\oldappendix\appendix
\renewcommand{\appendix}{%
  \oldappendix
  
  \patchcmd{\@sect}{\protected@edef\@currentlabel}{\protected@edef\@currentlabel\gdef\cref@currentlabel{appendix}}{}{}
  
  \renewcommand{\@seccntformat}[1]{%
    \ifstrequal{##1}{section}{Appendix\ \csname the##1\endcsname.\;}{\csname the##1\endcsname.\;}%
  }
}
\makeatother

\usepackage{algorithm}
\usepackage{algpseudocode}

\renewcommand{\tilde}{\widetilde}
\renewcommand{\hat}{\widehat}

\title{Impulse Response Estimation via Laguerre-Fourier Expansion}

\author[$\ast,1$]{Tamás~Dózsa}
\author[$\dagger,2$]{Art~J.~R.~Pelling}
\author[$\ast,3$]{Matthias~Voigt}
\affil[$\ast$]{Faculty of Mathematics and Computer Science, UniDistance Suisse, Brig, Switzerland.}
\affil[$\dagger$]{Department of Engineering Acoustics, Technische Universit\"at Berlin, Berlin, Germany.}
\affil[$1$]{Corresponding author, \email{tamas.dozsa@unidistance.ch}, \orcid{0000-0003-0919-4385}}
\affil[$2$]{\email{a.pelling@tu-berlin.de}, \orcid{0000-0003-3228-6069}}
\affil[$3$]{\email{matthias.voigt@fernuni.ch}, \orcid{0000-0001-8491-1861}}

\shorttitle{Laguerre-ETFE}
\shortauthor{Dózsa, Pelling et al.}
\shortdate{\today}
\shortinstitute{UniDistance Suisse, TU Berlin}

\abstract{\input{00_abstract}}


\novelty{
\begin{itemize}
\input{00_highlights}
\end{itemize}
}

\begin{document}
\maketitle 
\input{01_introduction}

\input{02_preliminaries}
\input{03_method}
\input{04_results}
\input{05_conclusion}
\input{funding}
\section*{Acronyms}
\input{acronyms}

\appendix

\setcounter{section}{0}
\renewcommand{\thesection}{\Alph{section}}
\renewcommand{\thesubsection}{\thesection.\arabic{subsection}}
\renewcommand{\thesubsubsection}{\thesubsection.\arabic{subsubsection}}
\counterwithin{equation}{section} 
\input{appendix}
%
\bibliographystyle{abbrvdoi}
\bibliography{references}
\end{document}

%% file: 00_abstract.tex
The \ac{etfe} is a widely used method for system identification of \ac{lti} systems in engineering disciplines such as acoustics, audio engineering, seismography, and tomography.
However, \ac{etfe} suffers from numerical limitations when the excitation signal is band-limited or vanishes at certain frequencies, which is a common physical constraint of the excitation in practice.
In such cases, division in the frequency domain becomes heavily ill-conditioned, and small measurement disturbances or numerical inaccuracies can degrade the solution.

This paper presents \acs{letfe}, a generalization of \ac{etfe} based on Laguerre-Fourier expansions that addresses these limitations.
After a suitable transformation, the method can yield a well-conditioned circulant problem even when the original \ac{etfe} system is ill-conditioned. Solving this problem via classical \ac{etfe} yields the discrete Laguerre-Fourier coefficients of the system's transfer function.
The desired \ac{ir} of the system-to-be-identified can then be recovered by a subsequent transformation pipeline.

We derive novel and efficient algorithms for performing these transformations and analyse the conditioning of the transformed problem, explicitly characterizing its dependence on the input and a parameter used in the Laguerre-Fourier expansion.
We evaluate the method on two simulated discrete-time \ac{lti} systems of varying complexity. The experiments demonstrate accurate \ac{ir} recovery for spectral-zero and band-limited excitation, where standard \ac{etfe} fails.

%% file: 00_highlights.tex
\item A generalization of \ac{etfe} using Laguerre-Fourier expansions that can improve conditioning for excitation signals that vanish at certain frequencies.
\item A novel and numerically efficient method for computing discrete Laguerre-Fourier coefficients of the impulse response using only scaling, shifting, and Fourier transforms.
\item A condition number bound characterizing the dependence of the Laguerre-Fourier coefficients on the excitation signal and the Laguerre parameter used in the transformation.
\item A numerically stable scheme for \ac{ir} recovery from Laguerre-Fourier coefficients using \ac{rkhs} properties.
\item Experimental validation on systems where standard \ac{etfe} fails.

%% file: 01_introduction.tex
\section{Introduction}
\label{sec:intro}
Many engineering identification tasks in acoustics, seismography, tomography, and signal processing
involve linear time-invariant (\ac{lti}) systems whose dynamics must be recovered from measured
input-output data. In this setting, a central task is the estimation of the system's \ac{ir}, since
the convolution of an input signal with the \ac{ir} fully characterizes the time-domain behaviour of
an \ac{lti} system. In the system-identification literature, this inverse problem is commonly
referred to as \emph{deconvolution}~\cite{hansen2002,ljung1999}.

A formal derivation of the deconvolution problem is given in \cref{sec:deconv}. For the present
discussion, it suffices to consider the recovery of an \ac{ir} \(h\in\C^{N-M+1}\), \(M,N\in\N\),
\(M\leq N\), from an input signal \(u\in\C^{M}\) and an output signal \(y\in\C^{N}\) such that
\(h\ast u \approx y\), where ``$\ast$" denotes the discrete convolution.

A well-established method for deconvolution is the \ac{etfe} method~\cite{ljung1985}. \ac{etfe} is
widely used in engineering practice, particularly in acoustics and audio engineering~\cite{farina2000,muller-trapet2020}. Its basic idea is to solve the deconvolution problem by division in the frequency
domain. This connection is derived formally in \cref{sec:problem}. Given the frequency response
\(U(\omega^k)\in\C\) of an excitation signal and the frequency response \(Y(\omega^k)\in\C\) of the
system output for \(\omega=\eu^{2\pi\iu/N}\), \(k=0,\,\dots,\,N-1\), the transfer function can be
estimated by
\begin{equation}
  \label{eq:spectral-division}
  H\big(\omega^k\big)=\frac{Y\big(\omega^k\big)}{U\big(\omega^k\big)}.
\end{equation}
The quantity \(H(\omega^k)\in\C\) is referred to as the \ac{etfe} because it estimates the transfer
function from measured input-output data, which are generally affected by noise and other error
signals. Nevertheless, the resulting \ac{etfe} can be shown to be an unbiased estimate of the true
solution~\cite{ljung1999}. Once \(H(\omega^k)\) is obtained, the \ac{ir} can be recovered by applying
the inverse Fourier transform. A prototype \ac{etfe} algorithm with suitable zero-padding to account
for the different signal lengths is given in \cref{alg:etfe}.
\begin{algorithm}[tb]
  \caption{\textsc{ETFE}($u,y$)\label{alg:etfe}}
  \begin{algorithmic}[1]
    \Require{Excitation signal \(u\in\mathbb{C}^M\) and response signal \(y\in\mathbb{C}^N\).}
    \Ensure{Truncated impulse response \(h\in\mathbb{C}^{N-M+1}\).} \State
    \(\tilde{u}\gets \begin{bmatrix}u_0&\cdots&u_{M-1}&0&\cdots&0\end{bmatrix}^\top\in\mathbb{C}^N.\)
    \Comment{Zero-pad the excitation signal.} \State \(U\gets \textsc{FFT}(\tilde{u})\),
    \(Y\gets \textsc{FFT}(y)\). \Comment{\acs{fft} of excitation and response.} \State
    \(H\gets Y\oslash U\in\C^N\). \Comment{Elementwise division.} \State
    \(\tilde{h}\gets\textsc{IFFT}(H/N)\). \Comment{Unnormalized inverse \acs{fft} of \(H/N\).}
    \State \(\begin{bmatrix}h&\ast\end{bmatrix}^\top\gets\tilde{h}\). \Comment{Discard the last
      \(M-1\) entries.}
  \end{algorithmic}
\end{algorithm}

Several practical problems arise when applying the \ac{etfe} method to identify \ac{lti} systems.
In particular, \ac{etfe} becomes unreliable whenever $U(\omega^k) \approx 0$, since the spectral
division in~\cref{eq:spectral-division} then amplifies numerical errors and measurement noise. This
is important in practice because many applications permit only band-limited excitation signals in
order to ensure safety and respect the physical constraints of the system; see, for example, the
modelling of nuclear power plants~\cite{soumelidis1991}. Moreover, most identification applications
consider discrete-time systems, so \ac{etfe} recovers the transfer function $H$ only at the
predetermined sampling points $\{\omega^k\}_{k=0}^{N-1} \subset \T$, where $\T$ denotes the
complex unit circle. Since these sampling points are usually spaced equidistantly, small values of
$U(\omega^k)$ cannot in general be avoided when information about high or low frequencies is sought
from band-limited excitation.

A possible remedy for these numerical issues is regularization~\cite{muller-trapet2020,marconato2017,fujimoto2020,hansen2002},
which uses prior knowledge about the input and system to add a regularizing term
\(\varepsilon(\omega^k)\in[0,\infty)\) wherever the input signal vanishes. The regularized input
\begin{equation*}
  U_\mathrm{reg}(\omega^k)\coloneqq\frac{|U(\omega^k)|^2+\varepsilon(\omega^k)}{\overline{U(\omega^k)}}
\end{equation*}
is then used in place of \(U(\omega^k)\) in~\cref{eq:spectral-division}. Although regularization can
mitigate the numerical difficulties of \cref{eq:spectral-division}, it also introduces bias and may
therefore reduce accuracy. In addition, the regularizer may contain free parameters whose optimal
choice can itself become a computational challenge~\cite{rebillat2018,pillonetto2014}. For these
reasons, \ac{etfe} is often avoided in favour of methods that recover alternative representations of
the system, including subspace methods~\cite{ljung1999,verhaegen1992,vanoverschee1994}, vector
fitting~\cite{gustavsen1999}, and pole-finding approaches~\cite{vandenhof2005}.

In this work, we propose a generalization of the \ac{etfe} method based on Laguerre-Fourier
expansions. More specifically, we represent the transfer function by its Laguerre-Fourier
coefficients and develop a procedure for estimating these coefficients from input-output data. We call this proposed generalization the \ac{letfe}.
Laguerre functions form an orthonormal basis in the Hilbert space that contains the transfer
function, and they depend on a free parameter that influences the convergence rate of the
corresponding expansion. Several earlier works~\cite{qian2011,tuma2019,guarnizo2018} study how
this parameter can be chosen to improve approximation quality, and the competitiveness of
identification schemes based on Laguerre approximation has already been
established~\cite{vandenhof2005,illg2024}. These structural properties enable an \ac{ir} recovery
scheme that remains applicable for band-limited inputs. In contrast to subspace methods, vector
fitting, and pole-finding approaches, this Laguerre-based method provides a direct estimation
framework that maintains the interpretability of the frequency-domain \ac{etfe} while addressing the
numerical limitations that motivate this work.

The most important contributions of this work are:
\begin{enumerate}
  \item Two algorithmic variants are proposed for computing discrete Laguerre-Fourier coefficients
        from time-domain input-output data by solving a linear system with circulant Toeplitz
        structure. The first method (\cref{alg:proto-laguerre-fourier}) directly solves this system,
        whereas the second (\cref{alg:laguerre-fourier}) is an efficient adaptation that uses only a
        single \ac{fft}. In the latter method, the discrete orthogonality of appropriately sampled
        Laguerre functions is exploited without explicit basis construction.
  \item An algorithm is proposed for impulse response recovery from Laguerre-Fourier coefficients,
        thereby completing the Laguerre \ac{etfe} pipeline; see~\cref{alg:recover-ir}.
  \item A condition number bound is derived for the circulant system (\cref{thm:condGamma}),
        explicitly characterizing its dependence on the input $U$ and the Laguerre parameter.
  \item The complete method is verified on two discrete-time \ac{lti} systems of varying complexity.
        The experiments show that the \ac{ir} can be recovered in cases where \ac{etfe} fails,
        namely when the excitation has spectral zeros or is band-limited; see~\cref{sec:exp}.
\end{enumerate}

The rest of this paper is organized as follows. \Cref{sec:preliminaries} reviews \ac{etfe} and the
mathematical background. \Cref{sec:methods} presents the proposed algorithms for recovering
Laguerre-Fourier coefficients and the corresponding \ac{ir}. \Cref{sec:exp} reports the
identification results for two linear dynamical systems. Finally, \cref{sec:conclusion} summarizes
the main findings and discusses future research directions.

%% file: 02_preliminaries.tex
\section{Preliminaries}
\label{sec:preliminaries}
This section develops the mathematical background for the proposed method in \cref{sec:methods}. It derives a finite-dimensional deconvolution model, reviews circulant matrices, and introduces discrete-time \ac{lti} systems and Laguerre-Fourier expansions.
\subsection{Deconvolution}
\label{sec:deconv}
We introduce the \emph{deconvolution problem} following \cite{kilmer1999,hansen2002}. In continuous time,
the convolution of two signals takes the form of an inhomogeneous Fredholm integral equation of the first kind
\begin{equation}
  \label{eq:fredholm}
  (h\ast u)(t)=\int_{-\infty}^{\infty}h(t-s)u(s)\,\mathrm{d}s=y(t),
\end{equation}
where \(u,h,y\in L_2(\R;\C)\) are the input signal, kernel, and output signal. 
To obtain a finite-dimensional model, assume that the signals are causal and supported on \([0,\tau]\), i.e., \(u(t)=y(t)=h(t)=0\) for \(t\notin[0,\tau]\). We then rewrite
\cref{eq:fredholm} as
\begin{equation*}
  \int_{0}^{\tau}h(t-s)u(s)\,\mathrm{d}s=y(t).
\end{equation*}
In practice, signals are measured at discrete-time instances \(\mathcal{T}_N=\{t_0,\,\dots,\,t_{N-1}\}\) with
\(0=t_0<t_1<\dots<t_{N-1}\). If \(\tau\leq t_{N-1}\) and the signals satisfy
\begin{align*}
  u(t)&=0\quad\text{for}\quad t> t_{M-1},\\
  h(t)&=0\quad\text{for}\quad t> t_{N-M},\\
  y(t)&=0\quad\text{for}\quad t> t_{N-1},
\end{align*}
for \(M\leq N\), the integral equation can be approximated by a quadrature rule with
abscissae \(\{t_0,\,\dots,\,t_{M-1}\}\subset\mathcal{T}_N\) and weights
\(\{w_0,\dots,w_{M-1}\} \subset \C\) to obtain
the following system of \(N\) linear equations
\begin{equation*}
  \begin{bmatrix}
    w_0h(t_0-t_0)&0&\cdots&0\\
    w_0h(t_1-t_0)&w_1h(t_1-t_1)&\ddots&\vdots\\
    \vdots&&\ddots&0\\
    w_0h(t_{M-1}-t_0)&\cdots&\cdots&w_{M-1}h(t_{M-1}-t_{M-1})\\
    \vdots&\ddots&&\vdots\\
    \vdots&&\ddots&\vdots\\
    w_0h(t_{N-1}-t_0)&\cdots&\cdots&w_{M-1}h(t_{N-1}-t_{M-1})\\    
  \end{bmatrix}
  \begin{bmatrix}
    u(t_0)\\\vdots\\u(t_{M-1})
  \end{bmatrix}=\begin{bmatrix}
    y(t_0)\\\vdots\\y(t_{N-1})
  \end{bmatrix}.
\end{equation*}
For equidistant sampling, let \(t_n=n t_\Delta\), \(n=0,\,\dots,\,N-1\), with \(t_\Delta>0\). We then pass to discrete time, absorb the quadrature weights into the discrete input, and set \(w_k=1\). Let \(\ell\) denote the vector space of unilateral complex sequences \((x_0,x_1,x_2,\ldots)\), and regard a sequence \(x\in\ell\) with \(x_n=0\) for \(n\geq N\) as its vector of first \(N\) entries in \(\C^N\).

Because convolution is commutative, \(h\ast u=u\ast h\), we write the discrete model with the input \(u\) as the convolution operator, that is, 
\begin{equation}
  \label{eq:T-conv}
  u\ast h=
  \underbrace{
\begin{bmatrix}
u_0     & 0         & \cdots  & \cdots  & \cdots  & \cdots  & 0 \\
u_1     & u_0     &    \ddots    &         &         &         & \vdots \\
\vdots  & u_1      & \ddots  &  \ddots       &         &         & \vdots \\
u_{M-1} & \vdots   & \ddots  & \ddots  &    \ddots     &         & \vdots \\
0       & u_{M-1}  &   & \ddots  & \ddots  &     \ddots    & \vdots \\
\vdots  &  \ddots     & \ddots  &    & \ddots  & \ddots  & 0 \\
\vdots  &         & \ddots  & \ddots  &   & \ddots  & u_0 \\
\vdots  &         &          & \ddots  & \ddots  &   & u_1 \\
\vdots  &         &         &         & \ddots  & u_{M-1} & \vdots \\
0       & \cdots  &\cdots  & \cdots  & \cdots  & 0       & u_{M-1}
\end{bmatrix}
  }_{ \eqqcolon \toeplitz \in\C^{N\times(N-M+1)}}
  \underbrace{\begin{bmatrix}
    h_0\\
    \vdots\\
    h_{N-M}
  \end{bmatrix}}_{ \eqqcolon h\in\C^{N-M+1}}
  =
  \underbrace{\begin{bmatrix}
    y_0\\
    \vdots\\
    y_{N-1}
  \end{bmatrix}}_{ \eqqcolon y\in\C^{N}}.
\end{equation}
The matrix \(\toeplitz\in\C^{N\times(N-M+1)}\) has constant diagonals and is therefore a Toeplitz matrix. Equation~\cref{eq:T-conv} is exact under the finite-support assumptions above. For the stable systems considered below, it serves as a finite observation window model when the impulse-response tail beyond index \(N-M\) is negligible at the measurement accuracy.
\subsection{Circulant Systems}
\label{sec:circulant}
A circulant matrix \(\circulant\in\C^{N\times N}\) is a Toeplitz matrix whose diagonals ``wrap around''~\cite[Chapter 4.8]{golub2013}, i.e.,
\begin{equation}
  \label{eq:circulant}
  \circulant =\begin{bmatrix}
    c_0&c_{N-1}&\cdots&c_1\\
    c_1&\ddots&\ddots&\vdots\\
    \vdots&\ddots&\ddots&c_{N-1}\\
    c_{N-1}&\cdots&c_1&c_0
  \end{bmatrix}\in\C^{N\times N}
\end{equation}
and it is completely determined by the entries of its first column, \(c_0,\,\dots,\,c_{N-1}\in\C\).
By introducing the \emph{downshift matrix}
\begin{equation}
  \label{eq:downshift}
  \downshift = \begin{bmatrix}
    0&\cdots&\cdots&0&1\\
    1&\ddots&&&0\\
    0&\ddots&\ddots&&\vdots\\
    \vdots &\ddots&\ddots&\ddots&\vdots\\
    0&\cdots&0&1&0
  \end{bmatrix}\in\C^{N\times N},
\end{equation}
which is a circulant permutation matrix, i.e., \(\downshift\downshift^\top=I_N\). The following properties are used later.
\begin{proposition}~\label{prop:circulant}
  \begin{enumerate}[label=(\roman*)]
    \item Any circulant matrix \(\circulant\in\C^{N\times N}\) determined by the first column
          \(\begin{bmatrix}c_0&\cdots&c_{N-1}\end{bmatrix}^\top\in\C^N\) can be expanded as a power
          series \(\circulant=\sum_{k=0}^{N-1}c_k\downshift^{k}\).
    \item The sum or product of two circulant matrices is a circulant matrix.
    \item Any two circulant matrices \(\circulant_1,\circulant_2\in\C^{N\times N}\) commute, i.e., \(\circulant_1\circulant_2=\circulant_2\circulant_1\).
  \end{enumerate}
\end{proposition}
More importantly, the eigenvalue decomposition of a circulant matrix has a rich structure that lends
itself to efficient computation. To this end, we define the Fourier (or DFT) matrix.
\begin{definition}[\emph{Fourier matrix}]
  For \(N\in\N\), denote
  \begin{equation}
    \label{eq:omega}
    \omega\coloneqq\eu^{\frac{2\pi\iu}{N}}\in\C.
  \end{equation}
  The \emph{Fourier matrix} of dimension $N \times N$ is defined as
  \begin{equation}
    \label{eq:Fourier}
    \fourier =\frac{1}{\sqrt{N}}\left[\overline{\omega}^{jk}\right]_{0\leq j,\,k\leq N-1}\in\C^{N\times N}.
  \end{equation}
\end{definition}
The Fourier matrix is unitary and diagonalizes the downshift matrix, according to the following lemma.
\begin{lemma}[{\cite[Lem. 4.8.1]{golub2013}}\label{lem:downshift}]
  It holds
  \begin{equation}
    \label{eq:D-evd}
    \fourier ^*\downshift \fourier =\operatorname{diag}\left(\omega^0,\,\dots,\,\omega^{N-1}\right).
  \end{equation}
\end{lemma}
This relationship yields the following theorem.
\begin{theorem}[{\cite[Thm. 4.8.2]{golub2013}}\label{thm:circulant}]
  The eigenvalue decomposition of a circulant matrix \(\circulant\in\C^{N\times N}\) as defined in
  \cref{eq:circulant} is given by
  \begin{equation}
    \label{eq:C-evd}
\circulant=\fourier ^*\operatorname{diag}(\mu_0,\,\dots,\,\mu_{N-1})\fourier ,
  \end{equation}
  where \(\fourier \in\C^{N\times N}\) is the Fourier matrix from \cref{eq:Fourier} and \(\mu_n\), $n=0,\,\ldots,\,N-1$
  denote the eigenvalues of \(\circulant\) which are given by the (scaled) Fourier coefficients of its first
  column \(c=\begin{bmatrix}c_0&\cdots&c_{N-1}\end{bmatrix}^\top\in\C^N\), i.e.,
  \(\begin{bmatrix}\mu_0&\cdots&\mu_{N-1}\end{bmatrix}^\top=\sqrt{N}\fourier c\).
\end{theorem}
Because a circulant matrix is Toeplitz, the circular convolution of two vectors \(x\in\C^N\) and
\(c\in\C^N\) can be written as a linear system of equations with operator \(\circulant\) as in
\cref{eq:circulant}, i.e.,
\begin{equation}
  \label{eq:C-conv-helper}
  x\ast c = \circulant x=y\in\C^N.
\end{equation}
If every entry of \(\fourier c\) is nonzero, plugging \cref{eq:C-evd} into \cref{eq:C-conv-helper} yields
\begin{align}
  \label{eq:etfe}
  &&x\ast c=\fourier ^*\operatorname{diag}(\mu_0,\dots,\mu_{N-1})\fourier x&=y\nonumber\\
  &\iff &\sqrt{N}\operatorname{diag}(\fourier c)\fourier x&=\fourier y\nonumber\\
  &\iff &x&=\frac{1}{\sqrt{N}}\fourier ^*\operatorname{diag}(\fourier c)^{-1}\fourier y,
\end{align}
which reveals a simple and fast strategy for computing the solution \(x\), as the multiplications
with the (inverse) Fourier matrix \(\fourier \) in \cref{eq:etfe} can be performed efficiently with an
\ac{fft} algorithm, see \cite[Algorithm 4.8.1]{golub2013}.
\begin{lemma}[Spectral condition number of circulant matrices\label{lem:C-cond}]
  Let \(\circulant\in\C^{N\times N}\) be a circulant matrix with first column \(c\in\C^N\). If \(\circulant\) is
  nonsingular, the spectral condition number of \(\circulant\) is given by
  \begin{equation*}
    \kappa_2(\circulant)={\lVert \circulant\rVert}_2\big\lVert \circulant^{-1}\big\rVert_2=\frac{\max_{k=0,\,\dots,N-1}\left|{[\fourier c]}_k\right|}{\min_{k=0,\,\dots,N-1}\left|{[\fourier c]}_k\right|}.
  \end{equation*}
\end{lemma}
\begin{proof}
  See \cref{sec:proof-C-cond}.
\end{proof}
\subsection{Discrete-Time LTI Systems}
\label{sec:systems}
A discrete-time \ac{siso} \ac{lti} system maps an input signal \(u\in\ell\) to an output signal \(y\in\ell\) via~\cite{vandenhof2005}
\begin{equation}
  \label{eq:conv}
  y = h \ast u,
\end{equation}
where $\ast : \ell \times \ell \to \ell$ denotes the discrete convolution operator. In
\cref{eq:conv}, $h \in \ell$ is referred to as the \ac{ir} of the system.
Consider the $\cZ$-transform defined as
\begin{equation}
  \label{eq:Ztrans}
  \cZ[x](z) \coloneqq \sum_{n=0}^{\infty} x_n z^n \quad (x \in \ell, z \in \C),
\end{equation}
whenever the sum exists. Let $\D \coloneqq \{z \in \C \,: \, |z| < 1 \}$ and
$\T \coloneqq \{z \in \C: \ |z| = 1 \}$ denote the open complex unit disk and its boundary. For a
causal and \ac{bibo}-stable system, equivalently \(h\in\ell_1\) (see \cref{sec:spaces}), applying the
$\cZ$-transform to \cref{eq:conv} gives
\begin{equation}
  \label{eq:freqdom}
  Y(z) = \tH(z)  U(z) \quad (z \in \overline{\D} \coloneqq \D \cup \T),
\end{equation}
where $Y \coloneqq \cZ[y], U \coloneqq \cZ[u]$, and $\tH \coloneqq \cZ[h]$. Here, $\tH$ is known
as the transfer function of the system. In our setting, transfer functions belong to
the Hardy space $H_2(\D)$ and can be expressed as
\begin{equation}
\label{eq:trfFull}
\tH(z) = z H(z),
\end{equation}
where
\begin{equation}
  \label{eq:transferfunc}
  H(z) \coloneqq \sum_{k=0}^{\ord-1} \frac{\res_k}{1 - \overline{\lambda}_k z},
\end{equation}
where $\ord \in \N$, $\res_k \in \C$, $\lambda_k \in \D$, and $z \in \overline{\D}$ for
 $k=0,\ldots,\ord-1$, see \cite{cdcAngino}. For a formal definition of the Hardy space \(H_2(\D)\), see \cref{sec:spaces}.
Thus, $H$ is a rational function fully defined by the residues $\res_k$ and the so-called mirror-image poles $\lambda_k \ (k=0,\ldots,\ord-1)$. The latter naming convention is appropriate, since
the poles of $H$ according to \cref{eq:transferfunc} coincide with
$1/\overline{\lambda}_k \ (k=0,\ldots,\ord-1)$, and thus they are the mirror-image reflections of
$\lambda_k$ across the boundary $\T$.
\begin{remark}
  Note that the transfer function definition given in \cref{eq:transferfunc} is different from the
  definition usually found in system theory literature~\cite{vandenhof2005}. Indeed, due to the nonstandard definition of the $\cZ$-transform with nonnegative exponents of $z$ instead of nonpositive ones, $\tH(1/z)$ is the standard transfer function typically used in systems and control. The methods proposed in this paper recover a sequence $h \in \ell_2$ that satisfies $\cZ[h] = H$, where $H$ is defined according to~\cref{eq:transferfunc}. For simplicity, we also refer to $H$ as the \emph{transfer function} and to $h$ as the \emph{impulse response} of the system. We note, however, that the true transfer function is defined by~\cref{eq:trfFull}. Consequently, the true \ac{ir} of the system that satisfies~\eqref{eq:conv} is given by $\tilde{h} \coloneqq (0, h_0, h_1, \ldots) \in \ell_2$. This also implies that the proposed methods can only identify impulse responses of systems with strictly proper transfer functions. For more details, see~\cref{app:sys}.
\end{remark}
Expanding \cref{eq:transferfunc} in powers of $z$ gives the \ac{ir} \(h = (h_0,h_1,h_2,\ldots)\) in \cref{eq:conv}:
\begin{equation}
  \label{eq:ir}
  h_n = \sum_{k=0}^{\ord-1} \res_k \overline{\lambda}_k^n \quad (n \in \N).
\end{equation}
Thus, the \ac{ir} is a finite sum of geometric sequences with quotients equal to the complex conjugates of the mirror-image poles. Since $|\lambda_k| < 1$, each term in \cref{eq:ir} decays as $n \to \infty$. The decay of $\{h_n\}_{n=0}^{\infty}$ can nevertheless be slow when $\ord$ or $|\res_k| \ (k=0,\ldots,\ord-1)$ are large, or when $|\lambda_k| \approx 1 \ (k=0,\ldots,\ord-1)$.

Finally, for the considered method it is important to note that $H_2(\D)$ is a \ac{rkhs}~\cite{Garcia2018} with the Szeg\H{o} kernel $\xi : \D \times \D \to \C$ defined as
\begin{equation}
  \label{eq:kernel}
  \xi(z, w) \coloneqq \frac{1}{1 - \overline{w}z} \quad (z,w \in \D).
\end{equation}
Note that for any $\lambda \in \D$, we have $\xi(\cdot, \lambda) \in H_2(\D)$. The Szeg\H{o} kernel satisfies the  \emph{reproducing property}
\begin{equation}
  \label{eq:reprod}
  {\langle \xi(\cdot, \lambda), f \rangle}_{H_2(\D)} = \overline{f(\lambda)} \quad (f \in H_2(\D), \ \lambda \in \D ),
\end{equation}
which we will use in our proposed method. For a proof of the above statements and a more profound discussion on Hardy spaces, we recommend~\cite{Garcia2018}.
\subsection{Laguerre-Fourier Expansions in \texorpdfstring{$H_2(\mathbb{D})$}{H2(D)}}
\label{sec:lag}
This subsection introduces the Laguerre functions used by the \ac{ir} recovery method in \cref{sec:irrec}.
\subsubsection{Laguerre-Fourier Coefficients}
Define the \emph{Blaschke factors} by
\begin{equation}
  \label{eq:Blaschfac}
  B^{a}(z) \coloneqq \frac{z-a}{1 - \overline{a}z} \quad (a \in \D, z \in \overline{\D}).
\end{equation}
For their properties, see \cref{sec:blaschke-factors}. The Laguerre functions are defined by
\begin{equation}
  \label{eq:Lag}
  L_n^{a}(z) \coloneqq \frac{\sqrt{1 - |a|^2}}{1 - \overline{a}z} B^a(z)^n \quad (a \in \D, z \in \overline{\D}, n \in \N).
\end{equation}
For any choice of $a \in \D$, these form a complete and orthonormal function system in $H_2(\D)$.
Thus, for any $H \in H_2(\D)$, we have
\begin{equation}
  \label{eq:LagExpanse}
  H = \sum_{n=0}^{\infty} {\langle H, L_n^a \rangle}_{H_2(\D)} L_n^a,
\end{equation}
where the convergence is understood in the $H_2(\D)$-norm. Since in a \ac{rkhs} (thus also in
$H_2(\D)$), convergence in norm implies pointwise convergence~\cite{aronszajn1950}, \cref{eq:LagExpanse} also holds
pointwise in $\D$. 

\begin{definition}[\emph{Laguerre-Fourier coefficients}]
  \label{def:lcf}
  Given $H \in H_2(\D)$ and $a \in \D$, for any $n \in \N$, the $n$-th \emph{Laguerre-Fourier coefficient}
  is defined as the number
  \begin{equation}
    \label{eq:lcoeff}
    \clag \coloneqq {\langle H, L_n^a \rangle}_{H_2(\D)}.
  \end{equation}
\end{definition}
\begin{remark}[Laguerre-Fourier coefficients generalize usual Fourier coefficients]
  Choosing $a=0$ yields $L_n^a(z) = z^n$ for $z\in \overline{\D}$. Using the definition of the
  $H_2(\D)$ inner product given in \cref{eq:H2inner}, we notice that if $a = 0$, then
  \begin{equation*}
    \clag = {\langle H, L_n^a \rangle}_{H_2(\D)} = \frac{1}{2 \pi} \int_{-\pi}^{\pi} H(\eu^{\iu t}) \eu^{-n \iu t}\, \du t.
  \end{equation*}
  Thus, for this special case, we obtain the well-known trigonometric Fourier coefficients.
\end{remark}

We can approximate $H \in H_2(\D)$ by projecting it onto an $n$-dimensional subspace of $H_2(\D)$ spanned by the first $n$ Laguerre functions. We call such a projection the $n$-th Laguerre-Fourier partial sum of $H$ and define it as
\begin{equation}
  \label{eq:partSum}
  S_n^a H \coloneqq  \sum_{k=0}^{n-1} h_k^a L_k^a \quad (a \in \D, n \in \N).
\end{equation}
Completeness of the Laguerre system gives $S_n^a H \to H$ in the $H_2(\D)$ norm as $n\to\infty$. If $H$ has the finite pole-residue structure in \cref{eq:transferfunc}, \cref{eq:reprod} gives
\begin{equation}
  \label{eq:lcoeffExplicit}
  \clag = {\langle H, L_n^a \rangle}_{H_2(\D)} = \sum_{k=0}^{\ord-1} \res_k \overline{L_n^a(\lambda_k)} = \sum_{k=0}^{\ord-1} \res_k \frac{\sqrt{1 - |a|^2}}{1 - a \overline{\lambda_k}} \overline{B^a(\lambda_k)}^n,
\end{equation}
where $\lambda_k \in \D \ (k=0,\ldots,\ord-1)$ denote the parameters of the transfer
function~\cref{eq:transferfunc}.
\subsubsection{Discrete Laguerre-Fourier Coefficients}
In our application, we are interested in estimating the Laguerre-Fourier coefficients of certain
$H_2(\D)$ functions sampled at $N \in \N$ points on $\T$. For this, we introduce the \emph{discrete} Laguerre-Fourier coefficients. Let $N \in \N$, $a \in \D$ and consider the discrete set
\begin{equation}
  \label{eq:discSet}
  \discSet \coloneqq \left\{ z \in \T \,:\, B^a(z)^N = 1 \right\}.
\end{equation}
Using
\cref{eq:Blaschfac} and the fact that Blaschke factors are self-maps on $\T$, we conclude that
$\big|\discSet\big| = N$. Since Blaschke factors are invertible, see \cref{eq:BlaschInv}, we
can easily compute the elements of $\discSet \coloneqq \{z_0,\ldots,z_{N-1} \}$ for given $a \in \D$ and $N \in \N$. Indeed, with
\(\omega\) as in \cref{eq:omega} it holds
\begin{equation}
  \label{eq:discSetExpl}
  B^a(z_k)^N = 1 \iff z_k = B^{-a}(\omega^k) \quad (k=0,\,\dots,\,N-1).
\end{equation}
By choosing an appropriate discrete measure, one can guarantee that the sampled Laguerre functions remain orthogonal with respect to the following discrete inner product of $H_2(\D)$-functions sampled on $\discSet$. Let
   \begin{align*}
       \hf &\coloneqq \begin{bmatrix} f(z_0) & \cdots & f(z_{N-1}) \end{bmatrix}^\top \in \C^N, \\
       \hg &\coloneqq \begin{bmatrix} g(z_0) & \cdots & g(z_{N-1}) \end{bmatrix}^\top \in \C^N.
   \end{align*}
Then we define
\begin{equation}
  \label{eq:discInner}
  {\big\langle \hf, \hg \big\rangle}_{\discSet} \coloneqq \sum_{k=0}^{N-1} f(z_k) \overline{g(z_k)} / \sigma(z_k)^2,
\end{equation}
where the discrete measure $\sigma$ is 
\begin{equation}
  \label{eq:discMes}
  \sigma(z)\coloneqq \frac{\sqrt{N(1 - |a|^2)}}{\left| 1 - \overline{a}z \right|} \quad (z \in \T).
\end{equation}
The following theorem makes the orthogonality of the sampled Laguerre functions in the discrete inner product explicit.
\begin{theorem}[Discrete orthogonality of sampled Laguerre functions]
  \label{thm:discorth}
  Let $N \in \N$, and let $\discSet$ be defined according to \cref{eq:discSet} with
  $\discSet \coloneqq \{z_0, z_1, \ldots z_{N-1} \} \subset \T$. Then, the vectors
  $\discLag_n$ for $n=0,\ldots,N-1$ with
  \begin{equation*}
    \discLag_n \coloneqq \begin{bmatrix}
      L_n^a(z_0) &
      L_n^a(z_1) &
      \cdots &
      L_n^a(z_{N-1})
    \end{bmatrix}^\top \in \C^N
  \end{equation*}
  form an orthogonal system with respect to the following discrete inner product on \(H_2(\D)\) given by~\cref{eq:discInner}.
\end{theorem}
We note that this result can be found in~\cite{fridli2020}; however, for clarity, we give a short proof in \cref{sec:proof-discorth}. Using the Laguerre functions sampled over the grid $\discSet$, we can define the discrete Laguerre-Fourier coefficients~\cite{dozsa2024} describing $H_2(\D)$ functions.
\begin{definition}[\emph{Discrete Laguerre-Fourier coefficients}]
  \label{def:dlcf}
  Let $N \in \N$, $0 \leq n \leq N-1$, and $a \in \D$. Further, let $L_n^a$ denote the $n$-th
  Laguerre function corresponding to the parameter $a$. For \(H\in H_2(\D)\) and its discretization 
  \begin{equation}\label{eq:discH}
    \discH \coloneqq \begin{bmatrix}
      H(z_0)&\cdots&H(z_{N-1})
    \end{bmatrix}^\top\in\C^N,
  \end{equation}
  on the discretization set $\discSet\coloneqq \{z_0, z_1, \ldots z_{N-1} \} \subset \T$,
  we call \(\dlag{h} = \begin{bmatrix} \dlag{h}[0] & \dots & \dlag{h}[N-1]\end{bmatrix}^\top\in\C^N\) 
  \begin{equation}
    \label{eq:discCoeffs}
    \dlag{h}[n] \coloneqq \big\langle \discH, \discLag_n \big\rangle_{\discSet} \quad (n=0,\ldots,N-1),
  \end{equation}
  the \emph{discrete Laguerre-Fourier coefficients} of \(H\), where the inner product
  $\langle \cdot, \cdot \rangle_{\discSet}$ is defined according to \cref{eq:discInner}.
\end{definition}
The following lemma establishes \cref{eq:discCoeffs} as a linear transformation and introduces
several useful properties that offer notational convenience and efficient numerical computation. Somewhat uncommon for Laguerre functions, we will follow a matrix formulation. This will make notation less heavy and facilitate the numerical analysis and efficient algorithmic implementation of the deconvolution problem later on. 
\begin{lemma}[Discrete Laguerre-Fourier expansion\label{lem:L-trans}]
  Let \(H\in H_2(\D)\) and consider its discretization \cref{eq:discH} on $\discSet$. Let \(\dlag{h}\) denote the vector
  of discrete Laguerre-Fourier coefficients with \(\dlag{h}[n]\) as given in \cref{eq:discCoeffs},
  and define
  \begin{equation*}
    \Sigma \coloneqq \operatorname{diag}\left(\sigma(z_0),\,\dots,\,\sigma(z_{N-1})\right)\in\C^{N\times N}
  \end{equation*}
  with the discrete measure \(\sigma\) given by \cref{eq:discMes}. Finally, let
  \begin{equation}
    \label{eq:Lmat}
    \L \coloneqq 
    \begin{bmatrix}
      \discLag_{0}&\cdots& \discLag_{N-1}
    \end{bmatrix}=
    \begin{bmatrix}
      L_0^a(z_0)&\cdots&L_{N-1}^a(z_0)\\\vdots&&\vdots\\
      L_0^a(z_{N-1})&\cdots&L_{N-1}^a(z_{N-1})
    \end{bmatrix}\in\C^{N\times N}
  \end{equation}
  denote the matrix of sampled Laguerre functions.
  Then, the following hold:
  \begin{enumerate}[label=(\roman*)]
    \item The matrix \(\Sigma^{-1}\L\) is unitary.
    \item The discrete Laguerre-Fourier coefficients can be computed by a linear transformation
          \begin{equation}
            \label{eq:lagtrans}
            \dlag{h} = (\L)^*\Sigma^{-2}\discH.
          \end{equation}
    \item The inverse transformation is given by
          \begin{equation}
            \label{eq:SLE}
            \discH=\L\dlag{h}.
          \end{equation}
    \item It holds
          \begin{equation}
            \label{eq:lag-F}
            \L=\sqrt{\frac{N}{1-|a|^2}}\fourier ^*(I_N+\overline{a}\downshift),
          \end{equation}
          where \(\fourier \) is the Fourier matrix from \cref{eq:Fourier} and \(\downshift\) is the downshift
          matrix from \cref{eq:downshift}.
    \item The discrete measure \(\sigma\) can be expressed in terms of the Fourier frequencies:
          \begin{equation*}
            \sigma(z_k)=\sqrt{\frac{N}{1-|a|^2}}\big|1+\overline{a}\omega^k\big|\quad (k=0,\,\dots,\,N-1).
          \end{equation*}
  \end{enumerate}
\end{lemma}
\begin{proof}
  We prove the statements in order:
  \begin{enumerate}[label=(\roman*)]
    \item The identity \(\Sigma^{-1}\L(\L)^*\Sigma^{-*}=I_N\) follows immediately from
          \cref{thm:discorth} and the fact that \(\Sigma\) is real.
    \item Considering the \(n\)-th entry of \((\L)^*\Sigma^{-2}\discH\), we obtain
          \begin{equation*}
            \sum_{k=0}^{N-1}H(z_k)\overline{L_n^a(z_k)} / \sigma(z_k)^2=\big\langle \discH, \discLag_n \big\rangle_{\discSet},
          \end{equation*}
          which is the \(n\)-th discrete Laguerre-Fourier coefficient \(\dlag{h}[n]\) by definition.
    \item Left-multiplying \cref{eq:lagtrans} with $\L$ and using the identity $\L(\L)^\ast = \Sigma^2$ from (i) leads to the desired result.
    \item Due to the definition of \(\omega\) and \(z_k\) in \cref{eq:omega} and
          \cref{eq:discSetExpl}, respectively, it holds that
          \(B^a(z_k)^n=B^a(B^{-a}(\omega^k))^n=\omega^{kn}\). By \cref{eq:Lag}, this implies that
          \begin{equation*}
            L_n^a(z_k)=\frac{\sqrt{1-|a|^2}}{1-\overline{a}z_k}\omega^{kn}.
          \end{equation*}
          With the above, and recalling that, by \cref{eq:Fourier}, it holds
          \([\fourier ]_{ij}=\overline{\omega}^{ij}/\sqrt{N}\), it is easily verified that
          \begin{equation}
            \label{eq:L-intermediate}
            \L=\sqrt{N}\sqrt{1-|a|^2}\operatorname{diag}(1-\overline{a}z_0,\,\dots,\,1-\overline{a}z_{N-1})^{-1}\fourier ^*.
          \end{equation}
          Let us now have a closer look at the term \(1-\overline{a}z_k\). Using the identity
          \cref{eq:elratAndBlasch}, an algebraic manipulation yields
          \begin{equation}
            \label{eq:zk-to-omegak}
            \begin{aligned}
              1-\overline{a}z_k=&1-\overline{a}B^a(B^{-a}(z_k)) = 1-\overline{a}B^{-a}(B^{a}(z_k)) \\
              \overset{\cref{eq:elratAndBlasch}}{=}&(1 - |a|^2) \xi(B^a(z_k), -a) = \frac{1-|a|^2}{1+\overline{a}\omega^k}.
            \end{aligned}
          \end{equation}
          Thus, \cref{eq:L-intermediate} is equivalent to
          \begin{align*}
            \L&=\sqrt{\frac{N}{1-|a|^2}}\operatorname{diag}\left(1+\overline{a}\omega^0,\,\dots,\,1+\overline{a}\omega^{N-1}\right)\fourier ^*\\
                   &=\sqrt{\frac{N}{1-|a|^2}}\left(I_N+\overline{a}\operatorname{diag}\left(\omega^0,\,\dots,\,\omega^{N-1}\right)\right)\fourier ^*.
          \end{align*}
          Finally, by \cref{lem:downshift}, it holds that
          \(\operatorname{diag}\left(\omega^0,\,\dots,\,\omega^{N-1}\right)\fourier ^*=\fourier ^*\downshift\) which
          proves the statement.
    \item The statement follows directly by plugging \cref{eq:zk-to-omegak} into \cref{eq:discMes}.
  \end{enumerate}
\end{proof}
Finally, we clarify the relationship between classical Laguerre-Fourier~\cref{eq:lcoeff} and
discrete orthogonal Laguerre-Fourier~\cref{eq:discCoeffs} coefficients. Let
$H(z) = (1 - \overline{\lambda}z)^{-1} \ (z \in \overline{\D}, \lambda \in \D)$. By
\cref{eq:reprod,eq:lcoeffExplicit}, the (continuous) Laguerre-Fourier coefficients defined in
\cref{eq:lcoeff} satisfy
\begin{equation*}
  \clag=\overline{L_n^a (\lambda)} \quad (a \in \D, n \in \N).
\end{equation*}
In contrast, the discrete Laguerre-Fourier coefficients of such an $H$ are given by~\cite{dozsa2024}
\begin{equation}
  \label{eq:discCoeffElrat}
  \dlag{h}[n]=\frac{\clag}{1 - \overline{B^a(\lambda)}^N}.
\end{equation}
Since $B^a : \D \to \D$ for any $a \in \D$, we have $|B^a(\lambda)| < 1$ in the denominator in
\cref{eq:discCoeffElrat}. Consequently, for each fixed $n$, $\dlag{h}[n] \to\clag$ as $N\to\infty$. Finally, we note
that if $N$ is large, we can compute the discretization set $\discSet$ efficiently using
\cref{eq:discSetExpl}. In addition, the components of the matrix $\L$ can be computed in a
numerically safe manner using the recursion formula for Laguerre functions $L_{n+1}^a(z) = B^a(z) \cdot L_n^a(z)$. However, as will be shown in the following, we can circumvent the construction of \(\L\) altogether in the proposed algorithms. 

%% file: 03_method.tex
\section{Methods}
\label{sec:methods}
\subsection{Circulant Reformulation of \ac{etfe}}
\label{sec:problem}
By an appropriate column completion of the Toeplitz matrix
\(\toeplitz \in \C^{N\times (N-M+1)}\) in \cref{eq:T-conv}, we obtain
\begin{equation}
  \label{eq:C-conv}
  u\ast h=
  \underbrace{
  \begin{bmatrix}
      u_0&0&\cdots&\cdots&0&u_{M-1}&\cdots&u_1\\
      u_1&u_0&\ddots&&\vdots&\ddots&\ddots&\vdots\\
      \vdots&&\ddots&\ddots&\vdots&&\ddots&u_{M-1}\\
      u_{M-1}&u_{M-2}&\cdots&u_0&0&\cdots&\cdots&0\\
      0&u_{M-1}&\ddots&&\ddots&\ddots&&\vdots\\
      \vdots&\ddots&\ddots&\ddots&&\ddots&\ddots&\vdots\\
      \vdots&&\ddots&\ddots&\ddots&&\ddots&0\\
      0&\cdots&\cdots&0&u_{M-1}&u_{M-2}&\cdots&u_0
    \end{bmatrix}
    }_{ \eqqcolon \tilde{\toeplitz}\in\C^{N\times N}}
  \underbrace{\begin{bmatrix}
    h_0\\\vdots\\h_{N-M}\\0\\\vdots\\0
  \end{bmatrix}}_{ \eqqcolon \tilde{h}\in\C^N}=y\in\C^N.
\end{equation}
Note that \cref{eq:C-conv} is equivalent to \cref{eq:T-conv} in the sense that \(\tilde{h}\), defined accordingly, is a
solution to \cref{eq:C-conv} if and only if \(h\in\C^{N-M+1}\) is a solution to \cref{eq:T-conv}. The completion of \(\toeplitz\)
yields the square circulant matrix \(\tilde{\toeplitz}\in\C^{N\times N}\). The main appeal of this reformulation is that
\cref{eq:C-conv} can be solved efficiently with \ac{fft} by \cite[Algorithm 4.8.1]{golub2013}, as outlined in
\cref{sec:circulant}.

In \cref{alg:etfe}, the zero-padding in line \(1\) performs the column completion, and lines \(2-4\)
implement \cref{eq:etfe}. The final step in line \(5\) extracts the solution to
match the initial dimensions in \cref{eq:T-conv}\footnote{It should be noted that the equivalence of approaches only holds in the case where the discarded trailing \(M-1\) entries of the solution \(\tilde{h}\) in \cref{alg:etfe} are equal to zero, as they are in \cref{eq:C-conv}. Generally, this cannot be guaranteed and non-zero trailing entries can lead to degradation of the solution. This is a drawback of the \ac{etfe} approach, i.e., solving the \emph{linear deconvolution} problem by a \emph{circular deconvolution} problem. We refer the reader to \cite[Sec. 5.6]{hansen2002} for a detailed description of this phenomenon.}. Numerically computing
\(\tilde{h}\) in \cref{eq:C-conv} via the eigenvalue decomposition of circulant matrices \cref{eq:etfe} is equivalent to
the prototype \ac{etfe} algorithm. A conceptual overview of the described technique is depicted in \cref{fig:etfe}. We
will use a similar schematic portrayal to visualize the approach of our proposed method in the following section.

As mentioned in \cref{sec:intro}, \ac{etfe} works well only when the input signal does not
approximately vanish at the Fourier frequencies \(\omega^k\), \(k=0,\dots,N-1\), for \(\omega\) as in
\cref{eq:omega}. By \cref{lem:C-cond}, the spectral condition number
\(\kappa_2(\tilde{\toeplitz})\) of \cref{eq:C-conv} deteriorates when \(U(\omega^k)\approx0\) for some
\(k\).
The following assumption formalizes this setting.
\begin{assumption}[Input with vanishing frequency information]
  The input signal \(\discu\in\C^N\) is such that \cref{eq:C-conv} is ill-conditioned, i.e.,
  \begin{equation*}
    \frac{\max_{k=0,\dots,N-1}\big|{[\fourier \discu]}_k\big|}{\min_{k=0,\dots,N-1}\big|{[\fourier \discu]}_k\big|}=\frac{\max_{k=0,\dots,N-1}\big|U(\omega^k)\big|}{\min_{k=0,\dots,N-1}\big|U(\omega^k)\big|}\gg1.
  \end{equation*}
  \label{ass:vanishing}
\end{assumption}
\begin{figure}[tb]
  \centering \includegraphics{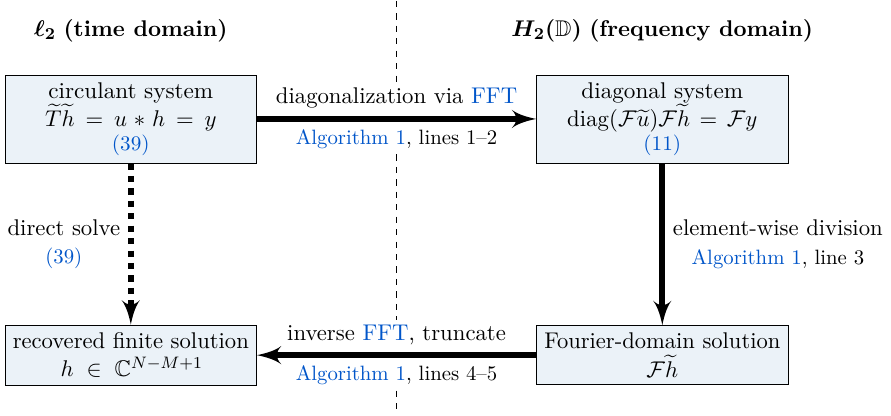}
  \caption{Schematic depiction of the \ac{etfe} algorithm as a commutative diagram.\label{fig:etfe}}
\end{figure}
\subsection{Proposed Method}
\label{sec:proposed-method}
If \(\tilde{u}\in\C^N\) satisfies \cref{ass:vanishing}, Fourier-domain division at the frequencies
\(\omega^k\), \(k=0,\ldots,N-1\), is numerically unreliable. We instead sample the truncated
\(\cZ\)-transform in \cref{eq:Ztrans} on the set \(\discSet\) from \cref{eq:discSet}, for some
\(a\in\D\), at frequencies \(z_k\in\discSet\), \(k=0,\ldots,N-1\). The parameter \(a\) changes the
sampling locations but does not guarantee that all sampled values of \(U\) are nonzero. In the following, denote
\begin{equation*}
  \discu_k \coloneqq \cZ\big[\tilde{u}\big](z_k), \quad \disch_k \coloneqq \cZ\big[\tilde{h}\big](z_k), \quad \discy_k \coloneqq \cZ[y](z_k)  \quad (k=0,\ldots,N-1)
\end{equation*}
as the truncated \(\cZ\)-transforms of the zero-padded input \(\tilde{u}\), \ac{ir} \(\tilde{h}\), and output \(y\),
respectively, sampled over \(\discSet\). The truncated \(\cZ\)-transform in \cref{eq:Ztrans} can also be written as a
matrix-vector multiplication with a Vandermonde matrix of the \(z_k\). Unlike \(\fourier\), a similarity
transformation with this Vandermonde matrix no longer diagonalizes \(\tilde{\toeplitz}\) in \cref{eq:C-conv}. This
means that we no longer make use of the rich circulant structure. Remarkably, we can employ the Laguerre-Fourier
transformation introduced in \cref{sec:lag} on the \(\cZ\)-transformed input, i.e., \(\L\Sigma^{-2}\discu\), to retain
the circulant structure of our problem. The following theorem, which is one of the two main results of this work,
reveals this connection in detail.
\begin{theorem}[Structure of the transformed problem\label{thm:shifted-lag}]
  Using the notation of the previous section, denote
  \begin{align*}
    \dlag{u}=(\L)^*\Sigma^{-2}\discu,\quad
    \dlag{y}=(\L)^*\Sigma^{-2}\discy,\quad
    \dlag{h}=(\L)^*\Sigma^{-2}\disch,
  \end{align*}
  as the discrete Laguerre-Fourier coefficients of \(\discu\), \(\discy\), and \(\disch\), respectively. Then, it holds
  \begin{equation}
    \label{eq:shifted-lag}
    \left((I_N+\overline{a}\downshift)\dlag{u}\right)\ast\dlag{h}=\sqrt{1-|a|^2}\dlag{y}.
  \end{equation}
\end{theorem}
\begin{proof}
  By \cref{eq:freqdom}, it clearly holds that \(U(z_k)H(z_k)=Y(z_k)\) for \(z_k\in\discSet\) and so
  \begin{align}
    &&\operatorname{diag}(\discu)\disch&=\discy\nonumber\\
    &\overset{\mathclap{\mathrm{\cref{lem:L-trans}(i)}}}{\iff}&(\L
                                                                (\L)^*\Sigma^{-2})\operatorname{diag}(\discu)(\L(\L)^*\Sigma^{-2})\disch&=\discy\nonumber\\
    &\iff&(\L)^*\Sigma^{-2}\operatorname{diag}(\discu)\L(\L)^*\Sigma^{-2}\disch&=(\L)^*\Sigma^{-2}\discy\nonumber\\
    &\iff&\underbrace{(\L)^*\Sigma^{-2}\operatorname{diag}(\discu)\L}_{\coloneqq\Gamma\in\C^{N\times N}}\dlag{h}&=\dlag{y}.\label{eq:Gamma}
  \end{align}
  Now, we take a closer look at \(\Gamma\) defined in \cref{eq:Gamma}. By \cref{lem:L-trans}(iv), it holds
  \begin{equation}
    \label{eq:Gamma1}
    \Gamma=\frac{1}{\sqrt{1-|a|^2}}\sqrt{N}(\L)^*\Sigma^{-2}\operatorname{diag}(\discu)\fourier ^*
    (I_N+\overline{a}\downshift).
  \end{equation}
  Denote
  \begin{equation}
      \label{eq:Phi1}
      \Phi\coloneqq\sqrt{N}(\L)^*\Sigma^{-2}\operatorname{diag}(\discu)\fourier ^*
  \end{equation}
  such that \(\Gamma=(1-|a|^2)^{-1/2}\Phi(I_N+\overline{a}\downshift)\). Applying \cref{lem:L-trans}(iv) again to \(\Phi\) yields
  \begin{equation}
    \label{eq:Phi}
    \Phi=\frac{N}{\sqrt{1-|a|^2}}(I+a\downshift^\top) \fourier \Sigma^{-2}\operatorname{diag}(\discu)\fourier^*.
  \end{equation}
  According to \cref{thm:circulant}, the matrix \(\fourier \Sigma^{-2}\operatorname{diag}(\discu)\fourier ^*\) in \cref{eq:Phi} is a
  circulant matrix and, by \cref{prop:circulant}(i), \((I_N+a\downshift^\top)\) is circulant as well. It follows that, as a
  product of circulant matrices, \(\Phi\) is circulant itself. Therefore, \(\Phi\) is determined by its first column
  \(\phi\coloneqq[\phi_0~\cdots~\phi_{N-1}]^\top\in\C^N\). Revisiting \cref{eq:Phi1}, the first
  column of \(\sqrt{N}\fourier ^*\) contains only ones. Hence,
  \begin{equation*}
    \sqrt{N}\operatorname{diag}(\discu)\fourier ^*
    = \begin{bmatrix}\discu & \ast & \cdots & \ast \end{bmatrix}
    \in \C^{N \times N}.
  \end{equation*}
  By \cref{eq:lagtrans}, the entries \(\phi_n\), \(n=0,\,\dots,\,N-1\), are given by
  \begin{equation}
    \label{eq:Phi1stCol}
    \phi_n=\sum_{k=0}^{N-1}\frac{U(z_k)\overline{L_n^a(z_k)}}{\sigma(z_k)^2}=\big\langle \discU, \discLag_n \big\rangle_{\discSet}=\dlag{u}[n],
  \end{equation}
  where $\discU$ denotes the discretization of $U$ on $\discSet$.
  The entries $\phi_n$, $n=0,\ldots,N-1$ are the discrete Laguerre-Fourier coefficients of the excitation by definition, so \(\phi=\dlag{u}\). By the
  same argument, \(\Gamma\) in \cref{eq:Gamma1} is circulant and determined by its first column, which is given by
  \((1-|a|^2)^{-1/2}\cdot(I_N+\overline{a}\downshift)\cdot\dlag{u}\). Substituting this into \cref{eq:Gamma} and reading
  the matrix multiplication as a convolution, we obtain \(\cref{eq:shifted-lag}\).
\end{proof}
The transformed equation \cref{eq:shifted-lag} is again a convolution equation. Thus, \(\dlag{h}\) can be computed
by solving a system with the same structure as \cref{eq:T-conv}, with
\((I_N+\overline{a}\downshift)\dlag{u}\) as input signal and
\(\sqrt{1-|a|^2}\dlag{y}\) as the right-hand side.
Assuming that \((I_N+\overline{a}\downshift)\dlag{u}\) now does not lead to ill-conditioning, we can solve the transformed
problem with classical \ac{etfe} to obtain the Laguerre-Fourier coefficients \(\dlag{h}\) of our desired solution \(h\).

We next investigate the numerical implications of \cref{thm:shifted-lag} and derive two
algorithms for recovering \(\dlag{h}\). The recovery of \(h\) from \(\dlag{h}\) follows in
\cref{sec:irrec}.
\subsubsection{Phase 1: Laguerre-Fourier Coefficient Recovery}
\label{sec:compute}
The key result of \cref{thm:shifted-lag} is that the matrix $\Gamma$ in \cref{eq:Gamma} is circulant. We can
therefore solve it with \ac{etfe}, i.e., \cref{alg:etfe}, using the Laguerre-Fourier coefficients of the response
signal \(\dlag{y}\) as right-hand side and the first column of \(\Gamma\) as ``excitation''. By
\cref{eq:Gamma1,eq:Phi1stCol}, the first column  of \(\Gamma\) is given by
\begin{equation}
  \label{eq:1stColGamma}
  \frac{1}{\sqrt{1 - |a|^2}}  (I_N+\overline{a}\downshift)\dlag{u} = \frac{1}{\sqrt{1 - |a|^2}}  \begin{bmatrix}
    \dlag{u}[0]&+&\overline{a}\dlag{u}[N-1] \\
    \dlag{u}[1]&+&\overline{a}\dlag{u}[0] \\
               &\vdots& \\
    \dlag{u}[N-1]&+&\overline{a}\dlag{u}[N-2]
  \end{bmatrix}.
\end{equation}
\begin{theorem}[Conditioning of $\Gamma$]
  \label{thm:condGamma}
  Let $\Gamma \in \C^{N \times N}$ be the circulant matrix defined according to \cref{eq:Gamma} whose first column
  \(\gamma\) is given by \cref{eq:1stColGamma}. If
  \begin{equation}
    \label{eq:wellCondCond}
    |\dlag{u}[0]| > \discs \coloneqq \sum_{k=1}^{N-1} |\dlag{u}[k]|,
  \end{equation}
  then
  \begin{equation*}
    \kappa_2(\Gamma) \leq \frac{1 + |a|}{1 - |a|} \cdot \frac{|\dlag{u}[0]| + \discs}{|\dlag{u}[0]| - \discs},
  \end{equation*}
  where $\dlag{u}[n] = \big\langle \discU, \discLag_n \big\rangle_{\discSet}$.
\end{theorem}
\begin{proof}
  The proof is given in~\cref{sec:proof-Gamma}.
\end{proof}
\cref{thm:condGamma} shows that the conditioning of the problem in \cref{eq:Gamma} mostly depends on how quickly
$|\dlag{u}[n]|$ tends to zero as $n$ increases. It also shows that choosing an $a \in \D$ close to the boundary makes the
problem ill-conditioned. These results show that the proposed method is expected to work well when
condition~\cref{eq:wellCondCond} is satisfied and $|a| \ll 1$. In such a case, \cref{eq:Gamma} can be solved
directly in a numerically safe manner, even if the input $U(z)$ vanishes in certain sampling points (which would cause
problems for the original \ac{etfe} method in~\cref{alg:etfe}). We note that exploiting the circulant nature of $\Gamma$
and \cref{eq:etfe}, \cref{eq:Gamma} can be solved without computing $\Gamma$ directly. Indeed, we have
\begin{equation}
  \label{eq:CircSys}
  \Gamma \dlag{h}=\dlag{y}\iff
  \dlag{h} = \frac{1}{\sqrt{N}}\fourier ^*\operatorname{diag}(\fourier  \gamma)^{-1}\fourier  \dlag{y},
\end{equation}
where $\gamma$ is the first column of $\Gamma$ in~\cref{eq:1stColGamma}, giving rise to
\cref{alg:proto-laguerre-fourier}.
%
%
%
%
The following prototype computes the Laguerre-Fourier coefficients.
\begin{algorithm}[tb]
\caption{\textsc{ComputeLaguerreFourierCoeffs\_Prototype}($u,y,a$)\label{alg:proto-laguerre-fourier}}
  \begin{algorithmic}[1]
    \Require{Excitation signal \(u\in\C^M\) and response \(y\in\C^N\), Laguerre parameter \(a\in\D\).}
    \Ensure{Discrete Laguerre coefficients \(\dlag{h}\in\C^{N}.\)} \State
    \(\tilde{u}\gets\begin{bmatrix}u_0&\cdots&u_{M-1}&0&\cdots&0\end{bmatrix}^\top\in\C^N\) \Comment{Zero-pad the excitation
      signal.}
      \State{$\zeta \gets \begin{bmatrix} B^{-a}\big(\omega^0\big) & \cdots & B^{-a}\big(\omega^{N-1}\big) \end{bmatrix}^\top \in (\discSet)^N$ as in \cref{eq:discSetExpl}}
      \State{\(\discu \gets \cZ[\tilde{u}](\zeta)\), \(\discy \gets \cZ[y](\zeta)\)} \Comment{Element-wise truncated \(\cZ\)-transforms \cref{eq:Ztrans}.} 
      \State \(\dlag{u}\gets (\L)^{*} \Sigma^{-2} \discu\),
      \(\dlag{y}\gets (\L)^{*} \Sigma^{-2} \discy\) \Comment{Compute discrete Laguerre coefficients.} 
      \For{\(k=0,\,\dots,\,N-1\)}\Comment{Generate the first column of $\Gamma$} 
      \State \(\gamma_k \gets (1 - |a|^2)^{-1/2}\cdot(\dlag{u}[k] + \overline{a} \dlag{u}[k-1 \ \textrm{mod} \ N]) \)
      \EndFor
    \State \(\dlag{h}\gets \textsc{ETFE}(\gamma, \dlag{y})\)
    \Comment{Apply \cref{alg:etfe}.}
  \end{algorithmic}
\end{algorithm}

The direct application of classical \ac{etfe} to \cref{eq:shifted-lag}, as in
\cref{alg:proto-laguerre-fourier}, is valid but can be made more efficient. Applying \cref{eq:etfe}
to \cref{eq:shifted-lag} yields
\begin{equation}
  \label{eq:hhat}
  \dlag{h}=\fourier ^*\operatorname{diag}(\underbrace{\fourier (I_N+\overline{a}\downshift)\dlag{u}}_{\eqqcolon \cX_1})^{-1}  \underbrace{\sqrt{\frac{1-|a|^2}{N}}\fourier \dlag{y}}_{\eqqcolon \cX_2}.
\end{equation}
Dissecting the individual terms in \cref{eq:hhat} gives
\begin{align}
    \mathcal{X}_1&=\fourier (I_N+\overline{a}\downshift)\dlag{u}\stackrel{\mathclap{\cref{eq:lagtrans}}}{=}\fourier (I_N+\overline{a}\downshift)(\L)^*\Sigma^{-2}\discu\nonumber\\
    &\stackrel{\mathclap{\cref{eq:lag-F}}}{=}\sqrt{\frac{N}{1-|a|^2}}\fourier (I_N+\overline{a}\downshift)(I_N+a\downshift^\top)\fourier \Sigma^{-2}\discu\nonumber\\
    &=\sqrt{\frac{N}{1-|a|^2}}\fourier ^2\fourier ^*(I_N+\overline{a}\downshift)\fourier \fourier ^*\big(I_N+a\downshift^\top\big)\fourier \Sigma^{-2}\discu\nonumber\\
    &\stackrel{\mathclap{\cref{eq:D-evd}}}{=}\sqrt{\frac{N}{1-|a|^2}}\fourier ^2\operatorname{diag}\left(\big|1+\overline{a}\omega^0\big|^2,\,\dots,\,\big|1+\overline{a}\omega^{N-1}\big|^2\right)\Sigma^{-2}\discu\nonumber\\
    &\stackrel{\mathclap{\cref{eq:zk-to-omegak}}}{=}\sqrt{\frac{1-|a|^2}{N}}\fourier ^2\discu,\label{eq:helper-u}
\end{align}
and similarly,
\begin{align}
    \mathcal{X}_2&=\sqrt{\frac{1-|a|^2}{N}}\fourier \dlag{y}\stackrel{\mathclap{\cref{eq:lagtrans}}}{=}\sqrt{\frac{1-|a|^2}{N}}\fourier (\L)^*\Sigma^{-2}\discy \nonumber \\ &\stackrel{\mathclap{\cref{eq:lag-F}}}{=}\fourier \big(I_N+a\downshift^\top\big)\fourier \Sigma^{-2}\discy\nonumber\\
    &=\fourier ^2\fourier ^*\big(I_N+a\downshift^\top\big)\fourier \Sigma^{-2}\discy\nonumber\\
    &\stackrel{\mathclap{\cref{eq:D-evd}}}{=}\fourier ^2\operatorname{diag}\big(1+a\overline{\omega}^{0},\,\dots,\,1+a\overline{\omega}^{N-1}\big)\Sigma^{-2}\discy\nonumber\\
    &\stackrel{\mathclap{\cref{eq:zk-to-omegak}}}{=}\frac{1-|a|^2}{N}\fourier ^2\operatorname{diag}(1+\overline{a}\omega^0,\,\dots,\,1+\overline{a}\omega^{N-1})^{-1}\discy\label{eq:helper-y}.
\end{align}
Substituting \(\mathcal{X}_1\) and \(\mathcal{X}_2\) in \cref{eq:hhat} with \cref{eq:helper-u,eq:helper-y} and pruning the scalar factors, we obtain
\begin{equation*}
    \dlag{h}=\sqrt{\frac{1-|a|^2}{N}}\fourier ^*\operatorname{diag}\big(\fourier ^2\discu\big)^{-1} \fourier ^2\operatorname{diag}\big(1+\overline{a}\omega^0,\,\dots,\,1+\overline{a}\omega^{N-1}\big)^{-1}\discy.
\end{equation*}
The matrix $\fourier ^2$ keeps the zeroth sample in place and reverses the order of all
remaining entries, i.e., $[\fourier ^2x]_k = x_{(N-k)\mod N}$ for all $x \in \C^N$.
Together with $\fourier ^4=I_N$~\cite{Candan2011}, this allows the reversal to be performed
after the element-wise division:
\begin{align}
    \dlag{h}&=  \sqrt{\frac{1-|a|^2}{N}}\fourier ^* \operatorname{diag}\big( \fourier ^2 \discu\big)^{-1} \fourier ^2 \operatorname{diag}\big(1+\overline{a}\omega^0,\,\dots,\,1+\overline{a}\omega^{N-1}\big)^{-1}\discy\nonumber \\ &= \sqrt{\frac{1-|a|^2}{N}}\fourier ^*\fourier ^2\operatorname{diag}(\discu)^{-1} \fourier ^4 \operatorname{diag}\big(1+\overline{a}\omega^0,\,\dots,\,1+\overline{a}\omega^{N-1}\big)^{-1}\discy\nonumber\\
    &=\sqrt{\frac{1-|a|^2}{N}}\fourier \operatorname{diag}(\discu)^{-1}\operatorname{diag}\big(1+\overline{a}\omega^0,\,\dots,\,1+\overline{a}\omega^{N-1}\big)^{-1}\discy. \label{eq:hhat-simplified}
\end{align}
The final statement \cref{eq:hhat-simplified} provides an efficient formula for computing $\dlag{h}$ that requires only a single \ac{fft}
call instead of the three \ac{fft} calls in \cref{alg:proto-laguerre-fourier}. This leads to \cref{alg:laguerre-fourier}.
\begin{algorithm}[tb]
  \caption{\textsc{ComputeLaguerreFourierCoeffs\_Efficient}($u,y,a$)\label{alg:laguerre-fourier}}
  \begin{algorithmic}[1]
    \Require{Excitation signal \(u\in\C^M\) and response \(y\in\C^N\), Laguerre parameter \(a\in\D\).}
    \Ensure{Discrete Laguerre-Fourier coefficients \(\dlag{h}\in\C^{N}.\)} 
    \State \(\tilde{u}\gets\begin{bmatrix}u_0&\cdots&u_{M-1}&0&\cdots&0\end{bmatrix}^\top\in\C^N\) \Comment{Zero-pad the excitation
      signal.} 
     \State{$\zeta \gets \begin{bmatrix} B^{-a}\big(\omega^0\big) & \cdots & B^{-a}\big(\omega^{N-1}\big) \end{bmatrix}^\top \in (\discSet)^N$ as in \cref{eq:discSetExpl}}
      \State{\(\discu \gets \cZ[\tilde{u}](\zeta)\), \(\discy \gets \cZ[y](\zeta)\)} \Comment{Element-wise truncated \(\cZ\)-transforms \cref{eq:Ztrans}.} 
    \State \(\hat{U} \gets\operatorname{diag}(1+\overline{a}\omega^0,\,\dots,\,1+\overline{a}\omega^{N-1})\discu\) \Comment{Scaling.} 
    \State
    \(\hat{H}_N\gets \discy\oslash\hat{U}\) \Comment{Elementwise division.} \State
    \(\dlag{h}\gets \textsc{FFT}(\hat{H}_N)\cdot \sqrt{(1-|a|^2)/N}\) \Comment{\ac{fft} of
      \(\hat{H}_N\).}
  \end{algorithmic}
\end{algorithm}

The element-wise division in step $5$ of \cref{alg:laguerre-fourier} remains numerically problematic when a sampled value $|U(z_k)|$ is small. The conditioning analysis above provides the alternative circulant solve in \cref{eq:CircSys} for this case. In the numerical experiments in \cref{sec:exp}, we did not observe numerical deterioration attributable to this efficient reformulation.

\subsubsection{Phase 2: Impulse Response Recovery}
\label{sec:irrec}
\begin{algorithm}
  \caption{\textsc{RecoverIR}($\dlag{h},a,N_q$)} 
  \label{alg:recover-ir}
  \begin{algorithmic}[1]
    \Require{Discrete Laguerre coefficients \(\dlag{h}\in\C^{N}\), Laguerre parameter \(a\in\D\), number of quadrature nodes $N_q \in \N$.}
    \Ensure{Truncated impulse response \(h \in \C^N\).}
 \State{$\rho \gets \begin{bmatrix} 1 & \eu^{\frac{2 \pi \iu}{N_q}} & \cdots & \eu^{\frac{2 \pi \iu (N_q-1)}{N_q}}  & 1 \end{bmatrix}^\top \in \T^{N_q+1}$}
    \State{\( \discH^a  \gets \cZ\big[\dlag{h}\big](\rho)\)} \Comment{scaled samples for the \(H_2(\D)\) inner product; use IFFT}
    \For{\(j=0,\,\dots,\,N-1\)}
    \For{\(k=0,\,\dots,\,N_q\)}
    \State \({[\varphi_j^a]}_k \gets \rho_k^j(1-\overline{a}\rho_k)/\sqrt{1-|a|^2}\)
    \EndFor
      \State \(h_j \gets \frac{1}{N_q + 1}{\sum_{k=0}^{N_q}}\discH^a_k{[\overline{\varphi_j^a}]}_k\)
    \EndFor
  \end{algorithmic}
\end{algorithm}
Suppose that, given $u,y \in \ell_2$ up to index $N-1 \in \N$, we have computed the discrete Laguerre-Fourier
coefficients \(\dlag{h}\) of the transfer function by solving \cref{eq:shifted-lag}. If the transfer function
$H \in H_2(\D)$ has the structure given in \cref{eq:transferfunc}, then by \cref{eq:lcoeffExplicit,eq:discCoeffElrat}
the discrete Laguerre-Fourier coefficients of \cref{def:dlcf} take the form
\begin{equation}
  \label{eq:transformedh}
  \dlag{h}[n]\coloneqq\sum_{k=0}^{\ord-1} \res_k \frac{\overline{L_n^a(\lambda_k)}}{1 - \overline{B^a(\lambda_k)}^N} = \\ \sum_{k=0}^{\ord-1} \res_k \frac{\sqrt{1 - |a|^2}}{1 - a \overline{\lambda}_k} \cdot \frac{\overline{B^a(\lambda_k)}^n}{1 - \overline{B^a(\lambda_k)}^N},
\end{equation}
where \(\res_k \in \C\), \(\lambda_k \in \D\), \(k=0,\ldots,\ord-1\), $N \in \N$, and $a \in \D$ are fixed and $0 \leq n \leq N-1$. Since $|B^a(\lambda)|^N \to 0$ exponentially, if $N \to \infty$, then for
any $\lambda \in \D$, the coefficients in \cref{eq:transformedh} approximate the Laguerre-Fourier coefficients presented
in \cref{eq:lcoeffExplicit}. If \(N\) is sufficiently large, we have $\overline{B^a(\lambda_k)}^N \approx 0$ for $k=0, \dots, \ord-1$. For this reason, the term
$\left(1 - \overline{B^a(\lambda_k)}^N\right)^{-1}$ in \cref{eq:transformedh} will be disregarded henceforth. However, note that if $|B^a(\lambda_k)| \approx 1$ for some $k$, then a very large $N$ may be necessary to achieve $\dlag{h}[n] \approx \clag$.

By \cref{eq:transformedh}, we have $\dlag{h}[n] \to \clag$ if $N \to \infty$, where $\clag$ are the Laguerre-Fourier
coefficients defined in \cref{eq:lcoeff}. Notice that by \cref{eq:lcoeffExplicit} the sequence $\clag$ can be written
as
\begin{equation}
  \label{eq:transIR}
  \clag = {\langle H, L_n^a \rangle}_{H_2(\D)} = \sum_{k=0}^{\ord-1} \trfRes_k \left(\trfPoles_k\right)^n,
\end{equation}
where $H$ is defined according to \cref{eq:transferfunc}, and 
\begin{equation*}
  \trfRes_k = \res_k \frac{\sqrt{1 - |a|^2}}{1 - a  \overline{\lambda}_k}, \quad 
  \trfPoles_k = \overline{B^a(\lambda_k)} \quad (k=0,\ldots,\ord-1).
\end{equation*}
Since $|\trfPoles_k| < 1$ for $k=0,\ldots,\ord-1$, \cref{eq:transIR} describes the $n$-th component of the \ac{ir} of a
\ac{bibo}-stable causal \ac{siso} \ac{lti} system characterized by the parameters $\{ \trfPoles_k \}_{k=0}^{\ord-1}$ and
the residues $\{ \trfRes_k \}_{k=0}^{\ord-1}$. The \ac{ir}
\begin{equation*}
 \trfIR \coloneqq \left(\clag \right)_{n \ge 0}
\end{equation*}
belongs to $\ell_2$. Therefore, the transformed transfer function
\begin{multline}
  \label{eq:transTRF}
  \trfTF(z) \coloneqq \cZ[\trfIR](z) = \sum_{k=0}^{\ord-1} \frac{\trfRes_k}{1 - \overline{\trfPoles}_k z} = \sum_{k=0}^{\ord-1} \trfRes_k \xi(z, B^a(\lambda_k)) \\ (z \in \overline{\D}, \trfPoles \in \D, k=0,\ldots,\ord-1),
\end{multline}
where $\xi$ denotes the Szeg\H{o} kernel from \cref{eq:kernel}, belongs to $H_2(\D)$. The \ac{ir} of the original system can be recovered from
the transformed transfer function in \cref{eq:transTRF} with the following theorem. The key idea behind our construction
is to exploit the \ac{rkhs} property (see \cref{eq:reprod}) of $H_2(\D)$ and construct a sequence of $H_2(\D)$ inner
products which coincide with the \ac{ir} $h$ of the system we are trying to identify.

\begin{theorem}[Impulse response recovery]
  \label{thm:irrec}
  Let $\{\lambda_k\}_{k=0}^{\ord-1} \subset \D$, $\{\res_k\}_{k=0}^{\ord-1} \subset \C$ for $\ord \in \N$ and consider a \ac{siso}
  \ac{lti} system described by the transfer function
  \begin{equation*}
    H(z) = \sum_{k=0}^{\ord-1} \frac{\res_k}{1 - \overline{\lambda}_k z} = \sum_{k=0}^{\ord-1} \res_k \xi(z, \lambda_k),
  \end{equation*}
  where $\xi$ is the Szeg\H{o} kernel given in \cref{eq:kernel}. Let $a \in \D$ and $\trfIR = (\trfIR_0,\trfIR_1,\ldots) \in \ell_2$ satisfy
  \begin{equation*}
    \trfIR_n \coloneqq {\langle H, L_n^a \rangle}_{H_2(\D)} \quad (n \in \N).
  \end{equation*}
  Define the transformed transfer function $\trfTF \coloneqq \cZ[\trfIR]$ as given in \cref{eq:transTRF}. Then, the \ac{ir}
  $h = (h_0,h_1,\ldots) = \cZ^{-1}[H]$ is given by
  \begin{equation}
    \label{eq:impresp}
    h_n = \sum_{k=0}^{\ord-1} \res_k \overline{\lambda}_k^n = \left\langle \trfTF, \varphi_n^a \right\rangle_{H_2(\D)},
  \end{equation}
  where
  \begin{equation}
    \label{eq:varphi}
    \varphi_n^a(z) \coloneqq \frac{(1 - \overline{a}B^{-a}(z))}{\sqrt{1 - |a|^2}}  B^{-a}(z)^n \quad (n \in \N, z \in \overline{\D}).
  \end{equation}
\end{theorem}
\begin{proof}
  By \cref{eq:transIR}, we have
  \begin{equation*}
    \trfIR_n = {\langle H, L_n^a \rangle}_{H_2(\D)} = \sum_{k=0}^{\ord-1} \res_k \frac{\sqrt{1 - |a|^2}}{1- a \overline{\lambda}_k} \overline{B^a(\lambda_k)}^n.
  \end{equation*}
  Hence, for any $z \in \D$,
  \begin{align*}
    \trfTF(z) &= \cZ[\trfIR](z) \\ &= \sum_{n=0}^{\infty} \sum_{k=0}^{\ord-1} \res_k \frac{\sqrt{1 - |a|^2}}{1- a \overline{\lambda}_k} \overline{B^a(\lambda_k)}^n  z^n \\ &= \sum_{k=0}^{\ord-1} \res_k \frac{\sqrt{1 - |a|^2}}{1- a \overline{\lambda}_k} \sum_{n=0}^{\infty}  \overline{B^a(\lambda_k)}^n  z^n \\ &=  \sum_{k=0}^{\ord-1} \res_k \frac{\sqrt{1 - |a|^2}}{1- a \overline{\lambda}_k} \cdot \frac{1}{1 - \overline{B^a(\lambda_k)}z} = \sum_{k=0}^{\ord-1} \res_k \frac{\sqrt{1 - |a|^2}}{1- a \overline{\lambda}_k} \xi(z, B^{a}(\lambda_k)).
  \end{align*}
  Finally,
  \begin{align*}
    \begin{split}
      \left\langle \trfTF, \varphi_n^a \right\rangle_{H_2(\D)} &\stackrel{\mathclap{\cref{eq:varphi}}}{=} \sum_{k=0}^{\ord-1} \res_k \frac{\sqrt{1 - |a|^2}}{1- a \overline{\lambda}_k} \left\langle \xi(\cdot, B^a(\lambda_k)), \frac{(1 - \overline{a}B^{-a}(\cdot))}{\sqrt{1 - |a|^2}} \cdot B^{-a}(\cdot)^n \right\rangle_{H_2(\D)} \\ 
      &= \sum_{k=0}^{\ord-1} \frac{\res_k}{1 - a \overline{\lambda}_k} \left\langle \xi(\cdot, B^a(\lambda_k)), (1 - \overline{a}B^{-a}(\cdot)) \cdot B^{-a}(\cdot)^n\right\rangle_{H_2(\D)} \\
      &\stackrel{\mathclap{\cref{eq:reprod}}}{=} \sum_{k=0}^{\ord-1} \frac{\res_k}{1 - a \overline{\lambda}_k}  \left(1 - a \overline{B^{-a}(B^{a}(\lambda_k))}\right) \cdot \overline{B^{-a}(B^{a}(\lambda_k))}^n \\ &\stackrel{\mathclap{\cref{eq:BlaschInv}}}{=} \sum_{k=0}^{\ord-1} \frac{\res_k}{1 - a \overline{\lambda}_k} (1 - a \overline{\lambda}_k) \overline{\lambda}_k^n = \sum_{k=0}^{\ord-1} \res_k \overline{\lambda}_k^n.
    \end{split}
  \end{align*}
\end{proof}

\begin{remark}
  \label{rem:reprod}
    In practice, the $H_2(\D)$ inner products from \cref{eq:impresp} can only be approximated. Notice, however,
          that (subject to hardware limitations) $\trfTF$ and a $\varphi_n^a$ can be evaluated over arbitrarily many points
          on $\T$. Therefore, we can approximate the inner products in \cref{eq:impresp} up to any precision using
          well-known quadratures. In this work, we use the composite trapezoid method, i.e.,
          \begin{equation}
          \label{eq:quadrature}
              h_n\approx \frac{1}{N_q+1} \sum_{k=0}^{N_q}{}\trfTF(\rho_k)\overline{\varphi_n^a(\rho_k)},\quad\rho_k\coloneqq\eu^{\frac{2\pi\iu k}{N_q}},
          \end{equation}
          for some discretization order \(N_q\in\N\). Due to the periodicity that arises from integrating along \(\T\), we weigh all quadrature nodes equally, instead of weighing the start and end nodes by a factor of 1/2. This means that \cref{eq:quadrature} is equivalent to the classical trapezoidal rule applied to the discretization with \(N_q+1\) points with identical start and end node \(q_0=q_{N_q}\). Also, note that the classical composite trapezoidal rule applied in this context gives
          \begin{equation*}
              \frac{2 \pi}{N_q + 1}\sum_{k=0}^{N_q} \trfTF(\rho_k)\overline{\varphi_n^a(\rho_k)} \approx \int_{0}^{2\pi} \trfTF\big(\eu^{\iu t}\big) \overline{\varphi_n^a\big(\eu^{\iu t}\big)}\, \du t =  2 \pi {\langle \trfTF, \varphi_n^a \rangle}_{H_2(\D)}.
          \end{equation*}
          Therefore~\cref{eq:quadrature} approximates the inner product ${\langle \trfTF, \varphi_n^a \rangle}_{H_2(\D)}$. It has been shown in \cite{trefethen2014} that \cref{eq:quadrature} is exponentially convergent on periodic contours. We observe this clearly in the numerical experiments in the following section.
\end{remark}
\begin{figure}[tbh]
  \centering \includegraphics{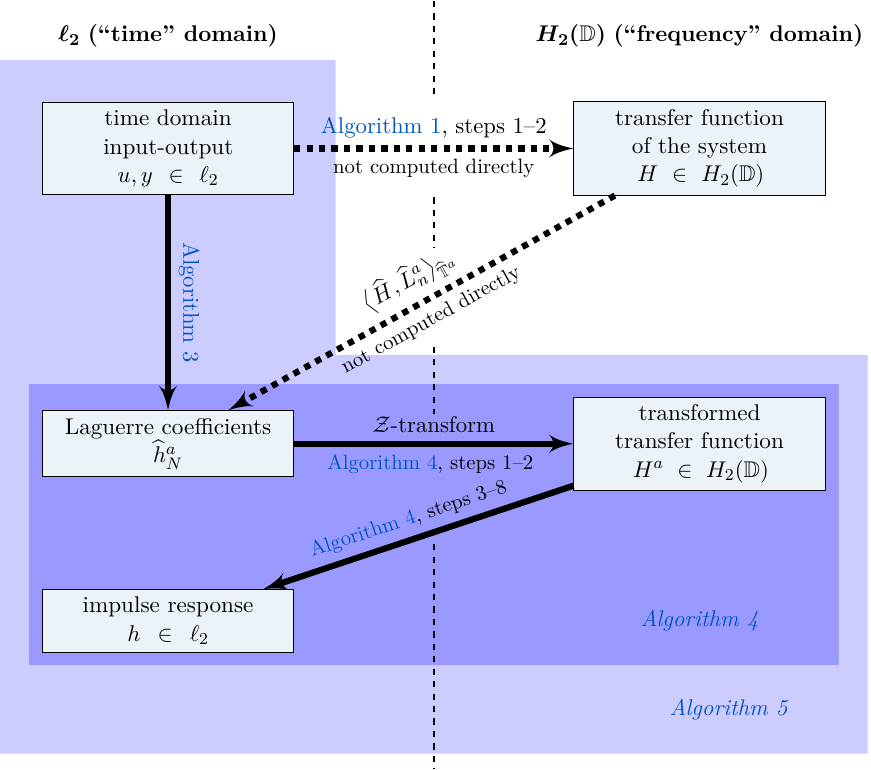}
  \caption{Schematic depiction of the proposed \ac{ir} recovery process.\label{fig:proc}}
\end{figure}
\begin{algorithm}[tbh]
\caption{\textsc{L-ETFE($u,y,a,N_q$)}\label{alg:lag-etfe}}
    \begin{algorithmic}[1]
        \Require{Excitation signal \(u\in\C^M\) and response signal \(y\in\C^N\), Laguerre parameter \(a\in\D\), number of quadrature nodes $N_q \in \N$.}
        \Ensure{Truncated impulse response \(h\in\C^{N}\).}
        \State \(\dlag{h}\gets\textsc{ComputeLaguerreFourierCoeffs\_Efficient}(u,y,a)\) \Comment{\cref{alg:laguerre-fourier}.}
        \State \(h\gets\textsc{RecoverIR}(\dlag{h},a,N_q)\) \Comment{\cref{alg:recover-ir}.}
    \end{algorithmic}
\end{algorithm}

\cref{alg:lag-etfe} combines the two phases and is referred to as \ac{letfe}. For chosen
parameters \(a\in\D\) and \(N_q\in\N\), it first computes \(\dlag{h}\) from the input-output data
\(u,y\) with \cref{alg:laguerre-fourier}. It then passes \(\dlag{h}\) to \cref{alg:recover-ir} to
return the first \(N\) samples of the recovered \ac{ir}. The prototype method in
\cref{alg:proto-laguerre-fourier} provides an alternative direct solve of \cref{eq:CircSys}. We did not
observe a case in our experiments in which it was numerically preferable to
\cref{alg:laguerre-fourier}, but this does not rule out such cases. \Cref{fig:proc} illustrates the
proposed transformation chain to recover the \ac{ir} from time-domain input and output signals.

%% file: 04_results.tex
\section{Results and Discussion}
\label{sec:exp}

\pgfplotstableread[col sep=comma]{data/exp1_scalars.csv}\expOneScalars
\pgfplotstableread[col sep=comma]{data/exp2_scalars.csv}\expTwoScalars
\pgfplotstableread[col sep=comma]{data/exp3_scalars.csv}\expThreeScalars
\pgfplotstableread[col sep=comma]{data/exp3_convergence.csv}\expThreeConvergence
\newcommand{\irestdefcell}[4]{%
  \pgfplotstablegetelem{#2}{#3}\of{#1}%
  \expandafter\def\expandafter#4\expandafter{\pgfplotsretval}%
}
\pgfplotstablegetelem{0}{runtime_s}\of{\expOneScalars}\edef\expOneEtfeRuntime{\pgfplotsretval}
\pgfplotstablegetelem{1}{runtime_s}\of{\expOneScalars}\edef\expOneLagAOneRuntime{\pgfplotsretval}
\pgfplotstablegetelem{2}{runtime_s}\of{\expOneScalars}\edef\expOneLagATwoRuntime{\pgfplotsretval}
\pgfplotstablegetelem{0}{cond_num}\of{\expOneScalars}\edef\expOneEtfeCondNum{\pgfplotsretval}
\pgfplotstablegetelem{1}{cond_num}\of{\expOneScalars}\edef\expOneLagAOneCondNum{\pgfplotsretval}
\pgfplotstablegetelem{2}{cond_num}\of{\expOneScalars}\edef\expOneLagATwoCondNum{\pgfplotsretval}
\def\lagOneParam{0.1 \iu }
\def\lagTwoParam{0.3 +  0.2 \iu}
\def\lagThreeParam{0.97 \cdot \eu^{ 0.052 \iu}}
\pgfplotstablegetelem{1}{runtime_s}\of{\expTwoScalars}\edef\expTwoLagAOneRuntime{\pgfplotsretval}
\pgfplotstablegetelem{2}{runtime_s}\of{\expTwoScalars}\edef\expTwoLagATwoRuntime{\pgfplotsretval}
\pgfplotstablegetelem{3}{runtime_s}\of{\expTwoScalars}\edef\expTwoLagAThreeRuntime{\pgfplotsretval}
\pgfplotstablegetelem{1}{cond_num}\of{\expTwoScalars}\edef\expTwoLagAOneCondNum{\pgfplotsretval}
\pgfplotstablegetelem{2}{cond_num}\of{\expTwoScalars}\edef\expTwoLagATwoCondNum{\pgfplotsretval}
\pgfplotstablegetelem{3}{cond_num}\of{\expTwoScalars}\edef\expTwoLagAThreeCondNum{\pgfplotsretval}
\pgfplotstablegetelem{0}{runtime_s}\of{\expThreeScalars}\edef\expThreeLagRuntime{\pgfplotsretval}
\pgfplotstablegetelem{0}{cond_num}\of{\expThreeScalars}\edef\expThreeLagCondNum{\pgfplotsretval}
\pgfplotstablegetelem{3}{relative_tail_h2}\of{\expThreeConvergence}\edef\expThreeTailAtThreeHundred{\pgfplotsretval}

This section evaluates the \ac{ir} recovery method proposed in \cref{sec:proposed-method} under random, spectral-zero, and band-limited excitation.

\subsection{Experimental Setup}

We estimate the \ac{ir} of two causal stable \ac{lti} systems with the proposed algorithm (see
\cref{fig:proc}) given the time-domain inputs and outputs $u,y \in \ell_2$. The first system is deliberately chosen as a very large
one, described by $\ord=20000$, while for the second experiment we choose a smaller system governed by $\ord=50$ mirror-image poles. For the implementation specifications of the systems, we refer to \cref{app:impl}. \Cref{fig:sys} illustrates the true \ac{ir} of each simulated system, and \cref{fig:poles} shows the corresponding mirror-image poles.
\begin{figure}[tb]
  \begin{subfigure}[t]{0.48\linewidth}
    \centering
    \includegraphics{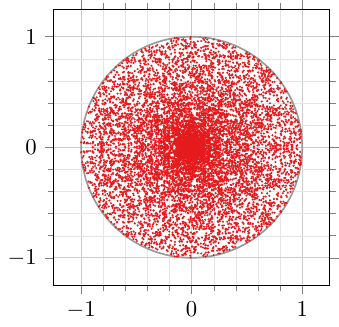}
    \subcaption{Large system ($\ord=20000$)}
  \end{subfigure}
  \hfill
  \begin{subfigure}[t]{0.48\linewidth}
    \centering
    \includegraphics{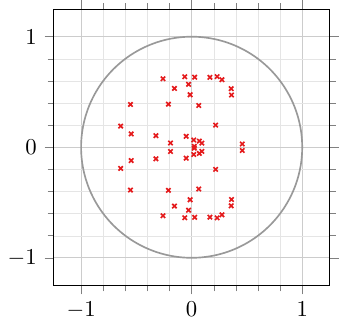}
    \subcaption{Small system ($\ord=50$)}
  \end{subfigure}
  \caption{Mirror image poles of the considered systems.}
  \label{fig:poles}
\end{figure}
\begin{figure}[tb]
  \begin{subfigure}[t]{0.5\linewidth}
    \centering
    \parbox[t][130pt][t]{\linewidth}{\vspace{0pt}\centering\includegraphics{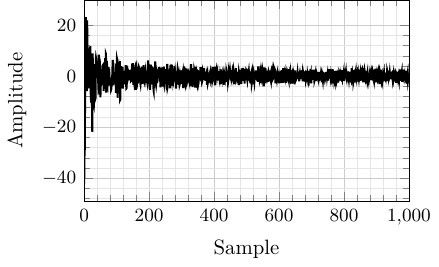}}
    \subcaption{Large system ($\ord=20000$)}
  \end{subfigure}
  \begin{subfigure}[t]{0.5\linewidth}
    \centering
    \parbox[t][130pt][t]{\linewidth}{\vspace{0pt}\centering\includegraphics{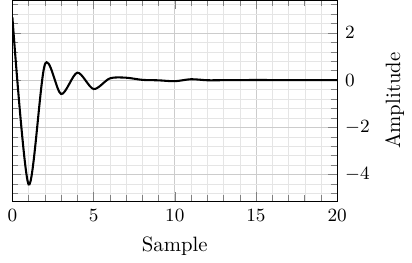}}
    \subcaption{Small system ($\ord=50$)}
  \end{subfigure}
  \caption{Impulse responses of the considered systems.}
  \label{fig:sys}
\end{figure}
In each experiment, we compute the pointwise difference between the true and recovered \ac{ir}. This error is given by
\begin{equation}
  \label{eq:err}
  \err \coloneqq \begin{bmatrix}h_0 - h^{\textrm{est}}_0 & \cdots & h_{N-1} - h^{\textrm{est}}_{N-1} \end{bmatrix}^\top \in \C^N,
\end{equation}
where $h$ and $h^{\textrm{est}}$ denote the true and estimated \ac{ir}, respectively. Based on~\cref{eq:err}, we also display the absolute and relative errors of the recovered \ac{ir} in the ${\| \cdot \|}_1$ and ${\| \cdot \|}_2$ vector norms for different choices of $N$. For a better visual and quantitative analysis, we list and display the errors in dB, i.e., \(20\operatorname{log}_{10}(\lVert \varepsilon\rVert)\), as is customary in engineering sciences.  Whenever applicable, we report the condition numbers of $\Gamma$ from~\cref{eq:Gamma} and the circulant \ac{etfe} system $\tilde{\toeplitz}$ from~\cref{eq:C-conv}. Finally, we record the wall-clock times of the algorithms (see also  \cref{app:impl} for the computer specification).
\subsection{Experiment 1: ETFE-Valid Random Excitation}
\label{subsec:largeLTI}

This experiment serves as a baseline on the larger system shown in \cref{fig:sys}; see also \cref{app:impl}. The excitation is chosen as a finite sequence of uniformly distributed random numbers in $[0,1)$. \ac{etfe} is expected to work well because $|U(z)|=|\cZ[u](z)| \gg 0$ on $\T$, so the division in~\cref{eq:etfe} is numerically safe. The purpose of this experiment is therefore not to show an advantage of the proposed method, but to test whether it can still reproduce the \ac{etfe} baseline and to show the effect of a nonzero Laguerre parameter.

The main result is that \ac{letfe} with $a=0$ reproduces the \ac{etfe} reconstruction throughout the reported range of $N$. This is visible in \cref{fig:randomInput}, where the pointwise errors coincide, and in \cref{tab:noise}, where the reported $\ell_1$-, $\ell_2$-, and relative errors are identical for all listed values of $N$. In contrast, choosing $a=\lagTwoParam$ leads to substantially larger errors for $N=10000$, $15000$, and $20000$.

This deterioration is consistent with the numerical properties of the method in the noise setting. The truncated $\cZ$-transform over $\discSet$ does not benefit from rapidly decaying Fourier coefficients of $U$, and the quadrature step in \cref{sec:irrec} introduces an additional approximation error for $a \neq 0$. Moreover, noise excitation is unfavorable for the conditioning result in \cref{thm:condGamma}, so poor choices of $a$ make the recovery of $\dlag{h}$ less reliable. For this reason, Experiment~1 should be read as a control case: when \ac{etfe} is applicable, the proposed method matches it for $a=0$, but it becomes sensitive to the Laguerre parameter away from the Fourier case.
\begin{figure}[tb]
  \begin{subfigure}[t]{0.48\linewidth}
    \centering
    \parbox[t][121pt][t]{\linewidth}{\vspace{0pt}\centering\includegraphics{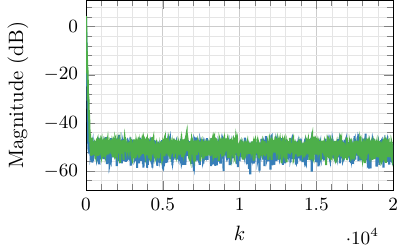}}
    \subcaption{Laguerre-Fourier coefficients.}
  \end{subfigure}
  \hfill
  \begin{subfigure}[t]{0.48\linewidth}
    \centering
    \parbox[t][121pt][t]{\linewidth}{\vspace{0pt}\centering\includegraphics{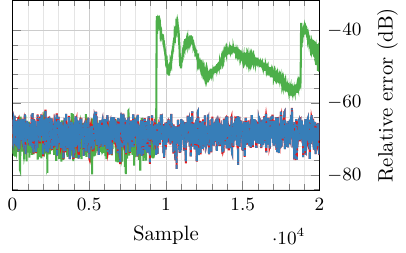}}
    \subcaption{Relative pointwise error.}
  \end{subfigure}
  \caption{Experiment~1: Laguerre-Fourier coefficient magnitudes and relative pointwise error in dB for random excitation. \irestplotkey{Set1-B} $a=0$; \irestplotkey{Set1-C} $a=\lagTwoParam$; \irestplotkey{Set1-A} \ac{etfe}. In~(b), the $a=0$ curve is dash-dotted. Curves are smoothed with a 16-sample moving average for visual clarity.}
  \label{fig:randomInput}
\end{figure}
\begin{table}[tb]
    \centering
    \caption{Recovery errors for random excitation using \ac{letfe} and \ac{etfe}.}
    \label{tab:noise}
    \begin{filecontents*}{data/exp1_table.csv}
method_key,method_tex,a_tex,N,l1,l2,l1rel,l2rel
lag_a1,\multirow{4}{*}{\shortstack[l]{\ac{letfe}\\$(a = 0.0 + 0.0 \iu)$}},0.0 + 0.0 \iu,100,0.3162798105,0.03934542222,0.008520477445,0.005633077791
lag_a1,,0.0 + 0.0 \iu,10000,27.83276896,0.3477233771,0.7341879977,0.04977578854
lag_a1,,0.0 + 0.0 \iu,15000,41.67869635,0.4252940812,1.099423441,0.06087985348
lag_a1,,0.0 + 0.0 \iu,20000,55.64824328,0.4915090481,1.467919788,0.0703583712
lag_a2,\multirow{4}{*}{\shortstack[l]{\ac{letfe}\\$(a = 0.3 + 0.2 \iu)$}},0.3 + 0.2 \iu,100,0.2116037067,0.02631371384,0.00570053652,0.003767330192
lag_a2,,0.3 + 0.2 \iu,10000,62.37591242,1.96689652,1.645385923,0.2815566387
lag_a2,,0.3 + 0.2 \iu,15000,230.2008668,3.532655842,6.072364331,0.5056914253
lag_a2,,0.3 + 0.2 \iu,20000,362.9225253,4.373518862,9.573368806,0.6260590008
etfe,\multirow{4}{*}{\ac{etfe}},--,100,0.3162798105,0.03934542222,0.008520477446,0.005633077791
etfe,,--,10000,27.83276896,0.3477233771,0.7341879977,0.04977578854
etfe,,--,15000,41.67869635,0.4252940812,1.099423441,0.06087985348
etfe,,--,20000,55.64824328,0.4915090481,1.467919788,0.0703583712
\end{filecontents*}
\pgfplotstabletypeset[irest-result-table-multi]{data/exp1_table.csv}
\end{table}
At $N=20000$, the proposed method required \irestprintfixed{\expOneLagAOneRuntime} seconds for $a=0$ and \irestprintfixed{\expOneLagATwoRuntime} seconds for $a=0.3+\iu 0.2$, whereas \ac{etfe} required only \irestprintsci{\expOneEtfeRuntime} seconds. This difference is expected because \ac{etfe} avoids using the trapezoidal rule as described in \cref{sec:irrec}. For $a=0$, $\kappa_2(\Gamma)=\irestprintfixed[2]{\expOneLagAOneCondNum}$ matches $\kappa_2(\tilde{\toeplitz})=\irestprintfixed[2]{\expOneEtfeCondNum}$. For $a=0.3+\iu 0.2$, $\kappa_2(\Gamma)=\irestprintfixed[2]{\expOneLagATwoCondNum}$.

\Cref{fig:runtimes} compares runtimes over different values of $N$. Both implementations use the efficient \ac{letfe} variant from \cref{alg:laguerre-fourier}. Although the proposed method is consistently slower than \ac{etfe}, this baseline isolates the main computational trade-off: in a regime where \ac{etfe} is already stable, \ac{letfe} offers comparable accuracy at $a=0$ but at higher cost.

\begin{figure}[tb]
    \begin{subfigure}[t]{0.48\linewidth}
    \centering
    \includegraphics{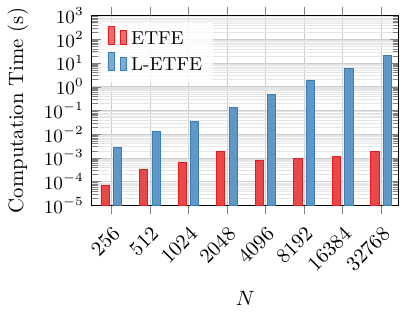}
    \subcaption{MATLAB implementation.}
  \end{subfigure}
  \hfill
   \begin{subfigure}[t]{0.48\linewidth}
    \centering
    \includegraphics{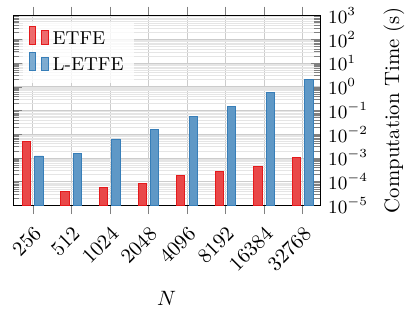}
    \subcaption{Python implementation.}
  \end{subfigure}
  \caption{Runtime comparisons of \ac{etfe} and \ac{letfe} with \(a=\lagTwoParam\).}
  \label{fig:runtimes}
\end{figure}
\subsection{Experiment 2: A Spectral Zero}
Next, we examine a scenario where \ac{etfe} fails, but the proposed methodology can successfully recover the \ac{ir}. In particular, we consider an input signal $u$ for which
\begin{equation}
    \label{eq:expInputProps}
    \cZ[u] = U \in H_2(\D), \quad U(1) = 0.
\end{equation}
The latter property makes \ac{etfe} fail, since to obtain $H(1)$, by~\cref{eq:etfe} one would have to divide by $0$. In contrast, choosing an appropriate parameter $a \neq 0$, we can ensure $\discSet \cap \{1\} = \emptyset$, avoiding this issue in~\cref{alg:laguerre-fourier}. For the details of the input generation, see \cref{app:impl}.

Given the large system illustrated in~\cref{fig:sys} and an excitation $u$ satisfying~\cref{eq:expInputProps}, the proposed method can recover the \ac{ir} with high accuracy. \Cref{fig:clagExp} shows the Laguerre-Fourier coefficient magnitudes and pointwise error up to $N=2000$ for the three parameter choices $a=\lagOneParam$, $a=\lagTwoParam$, and $a=\lagThreeParam$.
The above experiment is practically relevant in cases where the input is generated by a feedback controller. For example, controllers associated with industrial process control (such as PID loops), flight control, and automotive control
systems have been documented to produce excitation signals with near-zero components at certain frequencies (see, e.g.,~\cite{he2020}). \Cref{tab:exp} shows the quantitative results obtained for this experiment. Among the tested values, the smallest reconstruction errors are obtained for $a=\lagOneParam$, followed by $a=\lagTwoParam$, whereas $a=\lagThreeParam$ performs worst. This ordering is consistent with the fact that $a=0$ is not admissible because of~\cref{eq:expInputProps}, but parameter choices closer to $0$ still yield better approximations. Considering the distribution of the mirror-image poles in~\cref{fig:poles}, this behaviour is not surprising. By~\cref{thm:irrec}, the error of the reconstruction depends on how well the first $N$ discrete Laguerre coefficients capture the behaviour of the transfer function $H$. It is known~\cite[Eq. (15)]{soumelidis2017} that if $H$ can be written according to~\cref{eq:transferfunc}, then
\begin{equation*}
    \left\| H -  S_N^a H \right\|_{H_2(\D)} \leq C \max_{ k=0,\ldots,\ord-1 } \left| B^a( \lambda_k) \right|^N \quad (C > 0),
\end{equation*}
where the Laguerre-Fourier partial sums $S_N^a$ are defined according to~\cref{eq:partSum} and $\lambda_k$ denote the parameters defining the system in~\cref{eq:transferfunc}. Hence, we find that we can estimate how well the computed Laguerre-Fourier coefficients capture the behaviour of $H$ by considering the maximum possible pseudo-hyperbolic distance $\max_{ k=0,\ldots,\ord-1 } \left| B^a( \lambda_k) \right|$ between the Laguerre parameter $a$ and the mirror-image poles defining the system. Since according to~\cref{fig:poles}, these parameters cover much of the disk, choosing an $a$ far away from $0$ increases this distance and in turn the \ac{ir} recovery error.

\begin{table}[tbh]
    \centering
    \caption{Recovery errors for excitation with a spectral zero using \ac{letfe}.}
    \label{tab:exp}
    \begin{filecontents*}{data/exp2_table.csv}
method_key,method_tex,a_tex,N,l1,l2,l1rel,l2rel
lag_a1,\multirow{4}{*}{\shortstack[l]{\ac{letfe}\\$(a = 0.0 + 0.1 \iu)$}},0.0 + 0.1 \iu,100,1.050321002,0.1149662166,0.001902421685,0.00131256119
lag_a1,,0.0 + 0.1 \iu,10000,5.510246066,0.2130735931,0.001063188734,0.001853959858
lag_a1,,0.0 + 0.1 \iu,15000,5.510269602,0.2130735931,0.001002220418,0.001851980762
lag_a1,,0.0 + 0.1 \iu,20000,2651.100417,53.9239566,0.4734363674,0.4686407637
lag_a2,\multirow{4}{*}{\shortstack[l]{\ac{letfe}\\$(a = 0.3 + 0.2 \iu)$}},0.3 + 0.2 \iu,100,9.413964949,1.040520182,0.01705129292,0.01187954556
lag_a2,,0.3 + 0.2 \iu,10000,779.6605059,21.25414026,0.1504336206,0.1849329253
lag_a2,,0.3 + 0.2 \iu,15000,4074.193789,62.08476102,0.7410236702,0.5396247433
lag_a2,,0.3 + 0.2 \iu,20000,5987.520248,71.15791003,1.069257814,0.618417108
lag_a3,\multirow{4}{*}{\shortstack[l]{\ac{letfe}\\$(a = 0.97 \cdot \eu^{\iu 0.052})$}},0.97 \cdot \eu^{\iu 0.052},100,179.4094589,22.29455535,0.3249601261,0.2545353667
lag_a3,,0.97 \cdot \eu^{\iu 0.052},10000,24696.09414,310.3739794,4.765052004,2.700573499
lag_a3,,0.97 \cdot \eu^{\iu 0.052},15000,37509.78337,383.0019775,6.822365056,3.328954488
lag_a3,,0.97 \cdot \eu^{\iu 0.052},20000,50072.88542,441.6730507,8.942069802,3.838479385
\end{filecontents*}
\pgfplotstabletypeset[irest-result-table-multi]{data/exp2_table.csv}
\end{table}
\begin{figure}[tb]
  \begin{subfigure}[t]{0.48\linewidth}
    \centering
    \parbox[t][123pt][t]{\linewidth}{\vspace{0pt}\centering\includegraphics{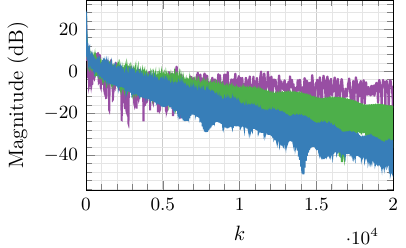}}
    \subcaption{Laguerre-Fourier coefficients.}
  \end{subfigure}
  \hfill
  \begin{subfigure}[t]{0.48\linewidth}
    \centering
    \parbox[t][123pt][t]{\linewidth}{\vspace{0pt}\centering\includegraphics{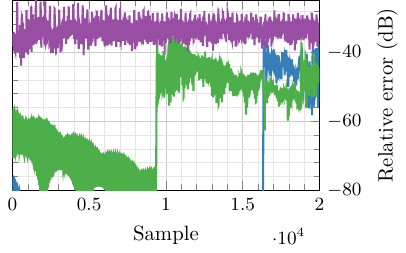}}
    \subcaption{Relative pointwise error.}
  \end{subfigure}
  \caption{Experiment~2: Laguerre-Fourier coefficient magnitudes and relative pointwise error in dB for spectral-zero excitation. \irestplotkey{Set1-B} $a=\lagOneParam$; \irestplotkey{Set1-C} $a=\lagTwoParam$; \irestplotkey{Set1-D} $a=\lagThreeParam$. Curves are smoothed with a 16-sample moving average for visual clarity.}
  \label{fig:clagExp}
\end{figure}

We emphasize that \ac{etfe} is omitted from~\cref{tab:exp}, because it does not produce a solution at all. The condition numbers of $\Gamma$ are \irestprintfixed[2]{\expTwoLagAOneCondNum}, \irestprintfixed[2]{\expTwoLagATwoCondNum}, and \irestprintfixed[2]{\expTwoLagAThreeCondNum} for $a = \lagOneParam$, $a = \lagTwoParam$, and $a = \lagThreeParam$, respectively. Respective run times are \irestprintfixed{\expTwoLagAOneRuntime} seconds for $a = \lagOneParam$, \irestprintfixed{\expTwoLagATwoRuntime} seconds for $a = \lagTwoParam$, and \irestprintfixed{\expTwoLagAThreeRuntime} seconds for $a = \lagThreeParam$.

\subsection{Experiment 3: Band-Limited Excitation}
\label{subsec:smallLTI}
In the next experiment, we consider the use of band-limited input signals. For this type of excitation, \ac{etfe} cannot
be used to recover the values of the transfer function at points that fall outside the support of $U$. Nevertheless,
band-limited inputs are frequently used in practice in cases where physical constraints on the system-to-be-identified
prevent the use of excitation signals with a rich frequency profile. Examples include the identification of aircraft
wing dynamics~\cite{Venkataraman2019} and nuclear power plants~\cite{soumelidis1991}. For our simulated experiment, we use the
excitation
\begin{equation}
  \label{eq:sinc}
  u_n = \frac{\omega_c}{\pi} \operatorname{sinc} \left( \frac{\omega_c}{\pi} n \right) = \begin{cases}
    \dfrac{\sin(\omega_c n)}{\pi n}, & n \neq 0, \\
    \dfrac{\omega_c}{\pi}, & n = 0,
  \end{cases}
\end{equation}
where $\omega_c \in [-\pi, \pi)$ is a cutoff frequency and $n \in \N$. The $\cZ$-transform of~\cref{eq:sinc} evaluated
over an equidistant grid on $\T$ is a rectangular pulse that vanishes for frequencies which fall outside the range
$[-\omega_c, \omega_c] \subset [-\pi, \pi]$. In our experiment the parameter $\omega_c$
satisfies
\begin{equation*}
  \omega_c = \frac{5 \pi}{6}.
\end{equation*}
\cref{fig:sincInp} illustrates the excitation signal in the time and frequency domains.
We consider the above input and measured output signals for $N=2000$ time steps. \Cref{fig:sampling} compares the equidistant sampling grid with the set $\discSet$, for three different nonzero choices of $a$ used in the experiments. The equidistant grid in \cref{fig:sampT} contains sampling points in the vanishing part of the input.
\DeclareRobustCommand{\samplingLegendSet}{\tikz[baseline=-0.6ex]{\path[draw=Set1-B, only marks, mark=x, mark size=3pt, line width=0.5pt, rotate=45] plot coordinates {(0,0)};}}
\DeclareRobustCommand{\samplingLegendParam}{\tikz[baseline=-0.6ex]{\path[index of colormap=0, only marks, mark=*, mark size=2pt] plot coordinates {(0,0)};}}
\DeclareRobustCommand{\samplingLegendZero}{\tikz[baseline=-0.6ex]{\draw[index of colormap=2, line width=1.2pt] (0,0) -- (0.9em,0);}}
\begin{figure}[tb]
  \begin{subfigure}[t]{0.48\linewidth}
    \centering
    \includegraphics{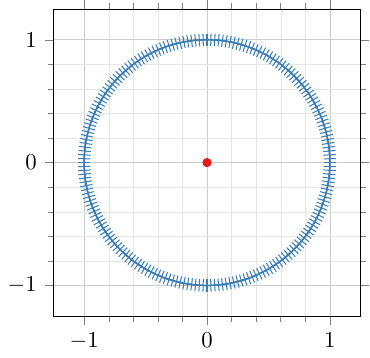}
    \subcaption{$a=0$. The equidistant grid samples the full circle.}
    \label{fig:sampT}
  \end{subfigure}
  \hfill
  \begin{subfigure}[t]{0.48\linewidth}
    \centering
    \includegraphics{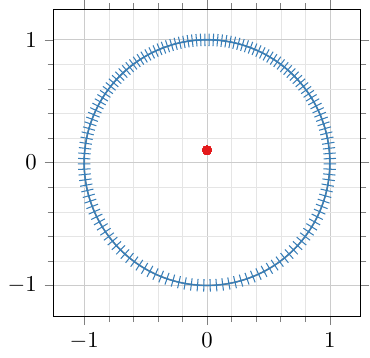}
    \subcaption{$a=\lagOneParam$. Mild clustering near $a$.}
    \label{fig:clusterAOne}
  \end{subfigure}

  \medskip

  \begin{subfigure}[t]{0.48\linewidth}
    \centering
    \includegraphics{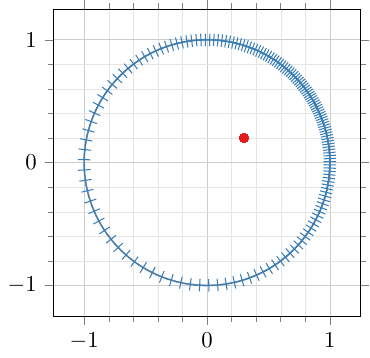}
    \subcaption{$a=\lagTwoParam$. Stronger clustering on one arc.}
    \label{fig:clusterATwo}
  \end{subfigure}
  \hfill
  \begin{subfigure}[t]{0.48\linewidth}
    \centering
    \includegraphics{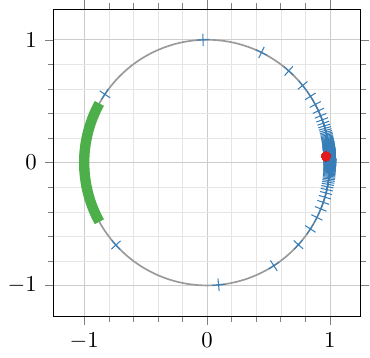}
    \subcaption{$a=\lagThreeParam$. All sampling points remain inside the excited region.}
    \label{fig:goodSamp}
  \end{subfigure}
  \caption{Samplings of the unit circle $\T$ for the sinc-input experiment. Legend: \samplingLegendSet\ $\discSet$, \samplingLegendParam\ $a$, \samplingLegendZero\ $U(z)=0$. From \subref{fig:sampT} to \subref{fig:goodSamp}, increasing $|a|$ shifts more sampling points toward the excited frequency range.}
  \label{fig:sampling}
\end{figure}

As $|a| \to 1$, most points $z_n \in \discSet$ cluster on an arc of $\T$ close to $a \in \D$; compare \cref{fig:clusterAOne}, \cref{fig:clusterATwo}, and \cref{fig:goodSamp}. This clustering is not perfect because some sampling points can still appear far from $a$, especially if $N$ is large. Choosing a moderate $N$, however, allows us to obtain a sampling grid $\discSet$ that falls completely inside the support of $U$ on $\T$, as shown in \cref{fig:goodSamp}.
Since our \ac{ir} recovery method in \cref{thm:irrec} uses $\trfTF = \cZ\big[\dlag{h}\big]$, our pipeline can only
approximate the \ac{ir} well if the transfer function is closely approximated by its $N$-th Laguerre-Fourier partial sum~\cref{eq:partSum}. That is, we need
\begin{equation}
    \label{eq:pSumApr}
    {\| H - S_N^a H \|}_{H_2(\D)} \approx 0
\end{equation}
to hold, where $S_N^a H$ is defined according to~\cref{eq:partSum} and $H$
denotes the transfer function of the system~\cref{eq:transferfunc}. By Parseval's theorem, this is equivalent to having $\big|\dlag{h}_n\big| \approx 0$ for $n \geq N$. Indeed, we have
\begin{align*}
    \| H - S_N^a H \|_{H_2(\D)}^2 &= \left\| \sum_{n=0}^{\infty} \left\langle H,  L_n^a\right\rangle_{H_2(\D)} L_n^a - \sum_{n=0}^{N-1} \left\langle H, L_n^a \right\rangle_{H_2(\D)} L_n^a  \right\|_{H_2(\D)}^2 \\&= \left\| \sum_{n=N}^{\infty} \left\langle H, L_n^a \right\rangle_{H_2(\D)} L_n^a \right\|_{H_2(\D)}^2 \stackrel{\mathrm{Parseval}}{=}\sum_{n=N}^{\infty} \big|\left\langle H, L_n^a \right\rangle_{H_2(\D)}\big|^2 \\&= \sum_{n=N}^{\infty} \left|\sum_{k=0}^{\ord-1} \res_k \frac{\sqrt{1 - |a|^2}}{1 - \overline{a} \lambda_k} B_a(\lambda_k)^n \right|^2,
\end{align*}
where $L_n^a$ are defined according to~\cref{eq:Lag}, and we use~\cref{eq:lcoeffExplicit}. We observe that the $n$-th Laguerre-Fourier coefficient is a linear combination of exactly $\ord$ convergent geometric sequences, since $\big| B_a(\lambda_k) \big| < 1$, whenever $|\lambda_k| < 1$. 

We conclude that the proposed \ac{ir} recovery method is only applicable whenever we can find $a \in \D$ such that $\left\langle H, L_n^a \right\rangle_{H_2(\D)}$ is negligible for $n > N$ and $N$ is not too large. If a system is defined by many parameters (e.g., the one considered in~\cref{subsec:largeLTI}), or if it has dominant poles close to $\T$, we expect a large $N$ for \cref{eq:pSumApr} to hold. In addition, for large $N$, we need $|a| \approx 1$ for the clustering phenomenon in~\cref{fig:goodSamp} to appear, which makes $\Gamma$ in~\cref{eq:Gamma} poorly conditioned. 

For these reasons, we consider the smaller system from~\cref{fig:sys}. The system is defined by $\ord=50$ mirror-image
poles. Choosing $a=\lagThreeParam$, we obtain the sampling grid $\discSet$ as
shown in~\cref{fig:goodSamp}. In~\cref{fig:SnErr}, we illustrate the discrete Laguerre-Fourier coefficient magnitudes and relative pointwise error for $n=0,\ldots,299$.
The smaller system has quickly decaying Laguerre-Fourier coefficients that satisfy~\cref{eq:pSumApr}; more precisely for $N=300$, \cref{eq:pSumApr} is equal to \irestprintfixed[4]{\expThreeTailAtThreeHundred} (about \(-30\,\)dB). Thus the partial sum satisfies~\cref{eq:pSumApr} to this measured truncation level, with the relative error defined in~\cref{eq:err}. We emphasize that \ac{etfe} cannot be used in this context, as it would not be able to provide any
information about the behaviour of the system over frequencies where $U(z) = 0$ (see~\cref{fig:goodSamp}). Hence, the
proposed method can capture the behaviour of the system along frequencies that were not excited by the input. This is also well reflected by~\cref{tab:sinc}, which depicts the obtained metrics for the experiment up to $N=300$. We note that to achieve the results in~\cref{tab:sinc} we zero-padded the computed $N=300$ coefficients $\dlag{h}$ before recovering the \ac{ir} with~\cref{alg:recover-ir} to $N_q = 20000$ (see remark following~\cref{thm:irrec}). This greatly improves the accuracy of the numerical quadratures used in~\cref{alg:recover-ir}. For this experiment, we observe a condition number $\kappa\left(\Gamma\right) = \irestprintfixed[4]{\expThreeLagCondNum}$ and a total runtime of \irestprintfixed{\expThreeLagRuntime} seconds.
\begin{table}[tbh]
    \centering
    \caption{Recovery errors for band-limited excitation using \ac{letfe}.}
    \label{tab:sinc}
    \begin{filecontents*}{data/exp3_table.csv}
method_key,method_tex,a_tex,N,l1,l2,l1rel,l2rel
lag,\multirow{4}{*}{\shortstack[l]{\ac{letfe}\\$(a = 0.97 \cdot \eu^{\iu 0.052})$}},0.97 \cdot \eu^{\iu 0.052},2,0.417522875,0.3227456543,0.05890942736,0.0625148508
lag,,0.97 \cdot \eu^{\iu 0.052},150,1.945981483,0.4375193199,0.2087326413,0.08313995705
lag,,0.97 \cdot \eu^{\iu 0.052},225,2.084513284,0.4380308199,0.2235920369,0.08323715524
lag,,0.97 \cdot \eu^{\iu 0.052},300,2.59270573,0.4729525594,0.2781025,0.08987318658
\end{filecontents*}
\pgfplotstabletypeset[irest-result-table-multi]{data/exp3_table.csv}
\end{table}
\begin{figure}[tb]
  \begin{subfigure}[t]{0.48\linewidth}
    \centering
    \parbox[t][126pt][t]{\linewidth}{\vspace*{-6pt}\centering\includegraphics{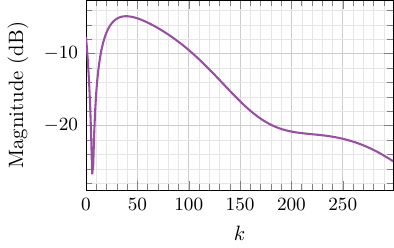}}
    \subcaption{Laguerre-Fourier coefficients.}
  \end{subfigure}
  \hfill
  \begin{subfigure}[t]{0.48\linewidth}
    \centering
    \parbox[t][126pt][t]{\linewidth}{\vspace*{-6pt}\centering\includegraphics{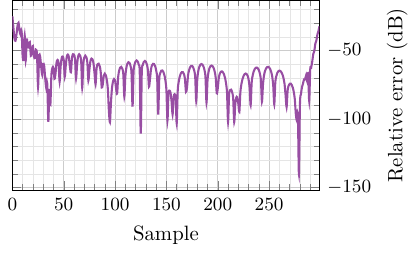}}
    \subcaption{Relative pointwise error.}
  \end{subfigure}
  \caption{Experiment~3: Laguerre-Fourier coefficient magnitudes and relative pointwise error in dB for band-limited excitation. \irestplotkey{Set1-D} $a=\lagThreeParam$.}
  \label{fig:SnErr}
\end{figure}

%% file: 05_conclusion.tex
\section{Conclusion}
\label{sec:conclusion}
In this study, we first propose a generalization of the \ac{etfe} method in \cite{ljung1985}. In particular, we
show that \ac{etfe} arises as a special case of~\cref{alg:proto-laguerre-fourier} and~\cref{alg:laguerre-fourier}. Using the proposed
generalized approach, however, allows us to recover the expansion coefficients of the transfer function even in cases,
where the input vanishes on the equidistant sampling grid of \ac{etfe}. Our second key contribution is~\cref{alg:recover-ir},
which allows us to recover the \ac{ir} given a sequence of Laguerre-Fourier coefficients. The complete pipeline is implemented in \cref{alg:lag-etfe}. In our experiments, we
demonstrate that the proposed methodology can be used to recover the \ac{ir} of very large \ac{lti} systems. In addition,
we show that our Laguerre-Fourier expansion-based method can be used to recover information about system behaviour over unexcited
frequencies, which can be particularly useful for identification scenarios where physical constraints on the system
only allow for the use of band-limited input signals.

In the next phase of our research, we plan to extend our method to \ac{mimo} systems. In addition, we would like to use
our identification pipeline to identify large acoustic systems appearing in real applications. Finally, based on
previous results such as~\cite{qian2011}, we would like to develop a method to automatically obtain the Laguerre
parameter appearing in~\cref{alg:proto-laguerre-fourier} and~\cref{alg:laguerre-fourier}.

%% file: funding.tex
\section*{CRediT Author Statement}
\textbf{Tam\'as D\'ozsa:} Conceptualization, Methodology, Software, Formal Analysis, Investigation, Writing - Original Draft, Visualization; 
\textbf{Art J. R. Pelling:} Conceptualization, Methodology, Software, Validation, Investigation, Data Curation, Writing - Original Draft, Visualization;
\textbf{Matthias Voigt:} Writing - Review \& Editing, Supervision, Funding Acquisition.

\section*{Acknowledgement}
We thank Alexandros Soumelidis for insightful discussions and early feedback over some of the technical details in the manuscript. 

\section*{Funding}
The work of Art J. R. Pelling was funded by the Deutsche Forschungsgemeinschaft (DFG, German
Research Foundation) – project no. \href{https://gepris.dfg.de/gepris/projekt/504367810?language=en}{504367810}. The research of Tam\'as D\'ozsa has received funding from the Swiss
Government Excellence Scholarship No. 2025.0057. The research of Matthias Voigt was funded in part by the Swiss National Science Foundation (SNSF) grant No. \href{https://data.snf.ch/grants/grant/224943}{224943}.

\section*{Code and Data Availability}
\begin{center}%
  \setlength{\fboxsep}{5pt}%
  \fbox{%
  \begin{minipage}{.92\linewidth}
    \textbf{Source code availability}\newline
    The source code and scripts used to compute the results presented in this
    paper can be obtained from
    \begin{center}
      \href{https://doi.org/10.5281/zenodo.21934090}%
        {\texttt{doi:10.5281/zenodo.21934090}}
    \end{center}
    under the MIT licence. The code was authored by Art J. R. Pelling and Tamás Dózsa.
  \end{minipage}}
\end{center}

%% file: acronyms.tex
\begin{acronym}[MIRACLE]\itemsep0pt
    \acro{bibo}[BIBO]{bounded-input bounded-output}
    \acro{blas}[BLAS]{Basic Linear Algebra Subprograms}
    \acro{dft}[DFT]{discrete Fourier transform}
    \acro{dof}[DOF]{degree of freedom}
    \acroplural{dof}[DOFs]{degrees of freedom}
    \acro{evd}[EVD]{eigenvalue decomposition}
    \acro{etfe}[ETFE]{empirical transfer function estimate}
    \acro{fir}[FIR]{finite impulse response}
    \acro{fft}[FFT]{fast Fourier transform}
    \acro{flop}[FLOP]{floating-point operation}
    \acro{fom}[FOM]{full-order model}
    \acro{hrtf}[HRTF]{head-related transfer function}
    \acro{iir}[IIR]{infinite impulse response}
    \acro{ifft}[IFFT]{inverse fast Fourier transform}
    \acro{ir}[IR]{impulse response}
    \acro{io}[IO]{input-output}
    \acro{lapack}[LAPACK]{Linear Algebra Package}
    \acro{letfe}[L-ETFE]{Laguerre empirical transfer function estimate}
    \acro{lti}[LTI]{linear time-invariant}
    \acro{mimo}[MIMO]{multiple-input multiple-output}
    \acro{miso}[MISO]{multiple-input single-output}
    \acro{mor}[MOR]{model order reduction}
    \acro{rir}[RIR]{room impulse response}
    \acro{rkhs}[RKHS]{reproducing kernel Hilbert space}
    \acro{rom}[ROM]{reduced-order model}
    \acro{simo}[SIMO]{single-input multiple-output}    
    \acro{siso}[SISO]{single-input single-output}
    \acro{svd}[SVD]{singular value decomposition}
\end{acronym}

%% file: appendix.tex
\section{Theoretical Background}
\label[appendix]{app:background}
\subsection{Signal and System Spaces}
\label{sec:spaces}
We make use of the following spaces of complex sequences:
\begin{align*}
  \ell_1 &\coloneqq \left\{ x \in \ell \,:\, {\|x\|}_{\ell_1} \coloneqq \sum_{n=0}^{\infty} |x_n| < \infty \right\}, \\
  \ell_2 &\coloneqq \left\{ x \in \ell \,:\, {\|x\|}_{\ell_2} \coloneqq \left(\sum_{n=0}^{\infty} |x_n|^2\right)^{1/2} < \infty \right\}.
\end{align*}
The latter forms a Hilbert space endowed with the inner product
\begin{equation*} {\langle x, u \rangle}_{\ell_2} \coloneqq \sum_{n=0}^{\infty} x_n \overline{u_n} \quad (x, u \in \ell_2).
\end{equation*}
If for any bounded input sequence $u$ the output sequence $y$ is also bounded, the system is called \ac{bibo}-stable.
This condition is equivalent to $h \in \ell_1$~\cite{vandenhof2005}. Consequently, the transfer function of
\ac{bibo}-stable systems satisfies $H \in H_{\infty}(\D)$, where $H_{\infty}(\D)$ denotes the space of holomorphic
functions bounded on $\D$. Then, it also belongs to the Hardy space $H_2(\D)$~\cite{vandenhof2005} defined as
\begin{equation*}
  H_2(\D) \coloneqq \left\{ f \in \fA(\D) \,:\, {\| f \|}_{H_2(\D)} \coloneqq \sup_{r < 1} \left( \frac{1}{2\pi} \int_{-\pi}^{\pi} \left| f\left(r  \eu^{\iu t}\right) \right|^2 \,\du t \right)^{1/2} < \infty \right\},
\end{equation*}
where $\fA(\D)$ denotes the space of all holomorphic functions on $\D$. Endowed with the inner product
\begin{equation}
  \label{eq:H2inner}
  {\langle f, g \rangle}_{H_2(\D)} \coloneqq \frac{1}{2 \pi} \int_{-\pi}^{\pi} f( \eu^{\iu t} ) \overline{g ( \eu^{\iu t} )} \,\du t,
\end{equation}
$H_2(\D)$ is a Hilbert space. In addition, since the point evaluation functional is bounded, $H_2(\D)$ is a reproducing kernel Hilbert space (\ac{rkhs}) with the kernel~\cref{eq:kernel}. For a deep discussion on \ac{rkhs}s we recommend~\cite{aronszajn1950}. 

\subsection{Blaschke Factors}
\label{sec:blaschke-factors}
Blaschke factors given in~\cref{eq:Blaschfac} are the building blocks of Laguerre functions and appear directly in the discretization set used in
\cref{sec:proposed-method}. They form a group with respect to function composition and are self-maps on $\T$ and $\D$.
The inverse of $B^a$ is given by
\begin{equation}
  \label{eq:BlaschInv}
  (B^{a})^{-1} = B^{-a} \quad (a \in \D).
\end{equation}
This property is heavily exploited in Theorem~\ref{thm:irrec}. We note that more general definitions exist for these objects; see, e.g.,~\cite{Garcia2018}. Finally, we recall the following property used in the proof of~\cref{lem:L-trans}:
\begin{lemma}
  Let $a \in \D$ and $z \in \overline{\D}$. Then, \begin{equation}
    \label{eq:elratAndBlasch}
    \left(1 - |a|^2\right) \cdot \xi(z, a) = \overline{a}B^{a}(z) + 1 
  \end{equation}
  where $\xi$ is the Szeg\H{o} kernel given in \cref{eq:kernel}.
\end{lemma}
\begin{proof}
  Indeed,
  \begin{align*}
    1 + \overline{a}B^a(z) &= 1 + \frac{\overline{a}z - |a|^2}{1 - \overline{a} z} = \frac{1 - \overline{a}z}{1 - \overline{a}z} + \frac{\overline{a}z - |a|^2}{1 - \overline{a} z} \\&= \frac{1 - \overline{a}z + \overline{a}z - |a|^2}{1 - \overline{a}z} =  \left(1 - |a|^2\right) \cdot \xi(z, a).
  \end{align*}
\end{proof}
\section{Considered System Class}
\label[appendix]{app:sys}
We discuss some important details about the considered class of systems. First, we show that the proposed method assumes the strict properness of the transfer function of the system to be identified. Let $\tH, H$ and $h = (h_0, h_1, \ldots)$ be defined according to~\cref{eq:trfFull},~\cref{eq:transferfunc} and~\cref{eq:ir}. Furthermore, define $\tilde{h} \in \ell_2$ as the sequence that satisfies $\cZ \big[ \tilde{h} \big] = \tH$. Then, by~\cref{eq:trfFull} and $\cZ[h] = H$, we have 
\begin{equation*}
    \tilde{h} \coloneqq \big(\tilde{h}_0, \tilde{h}_1,\tilde{h}_2, \ldots\big) = (0, h_0, h_1, \ldots).
\end{equation*}
Hence, we can only recover the \ac{ir} of a system if the first entry is $0$. This implies that the system's transfer function is strictly proper, which, due to our definition of the $\cZ$-transform~\cref{eq:Ztrans} means that $\tH(0) = 0$. Indeed,
\begin{equation*}
    \lim_{z \to 0} \tH(z) = \lim_{z \to 0} \cZ \big[\tilde{h} \big](z) \stackrel{\mathclap{\cref{eq:Ztrans}}}{=} \lim_{z \to 0} \sum_{n=0}^{\infty} \tilde{h}_n z^n = \tilde{h}_0 + \lim_{z \to 0} \sum_{n=1}^{\infty} \tilde{h}_n z^n = \tilde{h}_0.
\end{equation*}
Since $\tilde{h}_0 = 0$ by construction, this gives $\lim_{z \to 0} \tH(z) = 0$. On the other hand, $\tH$ is the rational transfer function of the original system and so
\begin{equation*}
    \tH(z) = \frac{B(z)}{A(z)} = \frac{b_0 + b_1 z + \cdots + b_m z^m}{a_0 + a_1 z + \cdots + a_n z^n}.
\end{equation*}
Under the assumption that $a_0 \neq 0$, evaluating the limit gives
\begin{equation*}
    \lim_{z \to 0} \tH(z) =  \lim_{z \to 0}  \frac{b_0 + b_1 z + \cdots + b_m z^m}{a_0 + a_1 z + \cdots + a_n z^n} = \frac{b_0}{a_0} = 0,
\end{equation*}
which implies $b_0 = 0$. If instead $a_0 = b_0 = 0$, then the numerator and denominator of $\tH$ can both be divided by $z$ and after a finite number of such reductions, one would end up with a rational function whose denominator polynomial is not homogeneous. This in turn implies the homogeneity of the corresponding numerator polynomial. This coincides with the transfer function being strictly proper. 

Finally, we discuss why~\cref{eq:transferfunc} holds when using the $\cZ$-transform with nonnegative exponents as defined in~\cref{eq:Ztrans}. The discrete-time system can be interpreted as an $\ell \to \ell$ operator mapping inputs to outputs. The effect of this operator on an input sequence $u \in \ell$ is given by~\cref{eq:conv}. It is well-known~\cite[Eq. (7.16) and Eq. (7.17)]{kailath1980} that the \ac{ir} of a discrete-time \ac{siso} \ac{lti} system satisfies~\cref{eq:ir}. Notice that in~\cref{eq:ir}, the parameters $\lambda_k$ belong to $\D$ for $k=0,\ldots,\ord-1$. This means, that applying the transformation~\cref{eq:Ztrans} to $h$ in~\cref{eq:ir} gives
\begin{equation*}
    H(z) = \cZ[h](z) = \sum_{n=0}^{\infty} h_n z^n = \sum_{n=0}^{\infty} \sum_{k=0}^{\ord-1} \res_k \overline{\lambda}_k^n z^n = \sum_{k=0}^{\ord-1} \res_k \sum_{n=0}^{\infty} \overline{\lambda_k}^n z^n.
\end{equation*}
Notice that if $z \in \overline{\D}$, then $\overline{\lambda}_k^n z^n \in \D$. Consequently, the infinite sum in the last equation is a geometric series. This gives~\cref{eq:transferfunc}. Finally, applying the $\cZ$-transform to obtain~\cref{eq:trfFull} gives
\begin{equation*}
   \tH(z) = \cZ \big[ \tilde{h} \big](z) = \sum_{n=0}^{\infty} \tilde{h}_n z^n = \sum_{n=1}^{\infty} h_{n-1} z^n = z \sum_{n=0}^{\infty} h_{n} z^n = z H(z).
\end{equation*}
\section{Proofs}
\label[appendix]{app:proofs}
\subsection{Proof of \texorpdfstring{\cref{thm:discorth}}{Theorem~\ref{thm:discorth}}}
\label{sec:proof-discorth}
Before proving \cref{thm:discorth}, we recall the so-called Christoffel-Darboux formula~\cite{vandenhof2005} (also known
as Džrbašjan's identity~\cite{dzrbasjan1962,lorentz1996}) summarized in the following theorem that will be used in the
proof.
\begin{theorem}[Christoffel-Darboux formula]
  \label{thm:CD}
  Let $a \in \D$, $z,w \in \overline{\D}$. Then,
  \begin{equation*}
    \sum_{n=0}^{N-1} L_n^a(z) \overline{L_n^a(w)} =
    \frac{1 - B^a(z)^N \overline{B^a(w)}^N}{1 - z\overline{w}},
  \end{equation*}
  where $L_n^a$ denotes the $n$-th Laguerre function defined in \cref{eq:Lag}.
\end{theorem}
\cref{thm:CD} can be proven by induction. Next, we prove \cref{thm:discorth}.
\begin{proof}
  Consider the matrix of sampled Laguerre functions \(\L\) as in \cref{eq:Lmat} and let \(\tL_{n}^a\) denote the
  \(n\)-th row of \(\L\). By \cref{thm:CD}, it holds
  \begin{equation}
    \label{eq:discOrthoInner}
    \left\langle \tL_{n}^a, \tL_{k}^a \right\rangle_{\C^N}=
    \sum_{j=0}^{N-1} L_n^a(z_j) \overline{L_k^a(z_j)} =
    \frac{1 - B^a(z_n)^N \overline{B^a(z_k)}^N}{1-z_n\overline{z_k}}.
  \end{equation}
  It follows from \cref{eq:discSet} that the numerator of \cref{eq:discOrthoInner} is zero for \(z_n\neq z_k\). Next,
  for \(z_n\rightarrow z_k\) and exploiting $z_k \in \discSet$, we obtain 
  \begin{align*}
    \lim_{z\rightarrow z_k}\frac{1 - B^a(z)^N \overline{B^a(z_k)}^N}{1 - z\overline{z_k}} &= \lim_{z\rightarrow z_k}\frac{1 - B^a(z)^N}{1-z\overline{z_k}} \\ &=\lim_{z\rightarrow z_k}\frac{\frac{\mathrm{d}}{\mathrm{d}z}(1 - B^a(z)^N)}{\frac{\mathrm{d}}{\mathrm{d}z}(1-z\overline{z_k})}\\ &=\lim_{z\rightarrow z_k}\frac{N\frac{1-|a|^2}{(1-\overline{a}z)^2}B^a(z)^{N-1}}{\overline{z_k}} \\
    &=\lim_{z\rightarrow z_k}N\frac{1-|a|^2}{\overline{z_k}(z-a)(1-\overline{a}z)}B^a(z)^N = N\frac{1-|a|^2}{|1-\overline{a}z_k|^2}=\sigma(z_k)^2,
  \end{align*}
  where the last equality of the first line is obtained by the rule of l'H\^{o}pital. With the above, it is easily
  verified that \(\L\operatorname{diag}(\sigma(z_0),\,\dots,\,\sigma(z_{N-1}))^{-1}\) is a unitary matrix which
  concludes the proof.
\end{proof}
\subsection{Proof of \texorpdfstring{\cref{lem:C-cond}}{Lemma~\ref{lem:C-cond}}}
\label{sec:proof-C-cond}
\begin{proof}
  Recall that the eigenvalues of $\circulant$ are given by \({[\fourier c]}_k=\mu_k\). By \cref{thm:circulant}
  \begin{align*}
    \circulant^*\circulant&=\fourier ^*\operatorname{diag}\left(\mu_1,\,\dots,\,\mu_{N}\right)\operatorname{diag}\left(\overline{\mu_1},\,\dots,\overline{\mu_{N}}\right)\fourier \\&=\fourier ^*\operatorname{diag}\left(|\mu_1|^2,\,\dots,\,|\mu_N|^2\right)\fourier .
  \end{align*}
  This implies that
  \begin{equation*}
    {\lVert \circulant\lVert}_2=\sigma_\mathrm{max}(\circulant)=\sqrt{\lambda_\mathrm{max}(\circulant^*\circulant)}=\sqrt{\max_{k=1,\,\dots,N}|\mu_k|^2}=\max_{k=1,\,\dots,N}|{[\fourier c]}_k|,
  \end{equation*}
  where $\sigma_{\max}$ and $\lambda_{\max}$ denote the maximum singular value and maximum eigenvalue of a (symmetric) matrix.
  Together with the connection \(\mu_k\big(\circulant^{-1}\big)=\mu_k^{-1}\), the proof is complete.
\end{proof}
\subsection{Proof of \texorpdfstring{\cref{thm:condGamma}}{Theorem~\ref{thm:condGamma}}}
\label{sec:proof-Gamma}
\begin{proof}
  Consider the first column $\gamma$ of $\Gamma$ as given in \cref{eq:1stColGamma}. By \cref{thm:circulant} we know that
  \begin{equation*}
    \mu_j \coloneqq \gamma_0 + \sum_{k=1}^{N-1} \gamma_k \omega^{kj} \quad (j=0,\ldots,N-1)
  \end{equation*}
  are the eigenvalues of $\Gamma$, where $\omega$ is defined according to \cref{eq:omega}. Noticing that according to \cref{eq:1stColGamma}, we have
  \begin{equation*}
  \gamma_k = \frac{1}{\sqrt{1-|a|^2}}\left(\dlag{u}[k] + \overline{a}\dlag{u}[k-1 \, \textrm{mod} \, N]\right).
  \end{equation*}
  Rearranging the sum, for $j=0,\ldots,N-1$
  we get
  \begin{equation*}
    \mu_j = \frac{1}{\sqrt{1-|a|^2}}\left( (1 + \overline{a}\omega^{j}) \dlag{u}[0] + \sum_{k=1}^{N-1} \left(\omega^{jk} + \overline{a} \omega^{j ( k+1 \, \textrm{mod} \, N)}\right) \dlag{u}[k]\right).
  \end{equation*}
  Notice that by the triangle and reverse triangle inequalities, for all $0 \le j,k < N-1$, we have 
  \begin{equation*}
    0 < 1 - |a| \leq \left|\omega^{jk} + \overline{a}\omega^{j (k+1 \, \textrm{mod} \, N)}\right| \leq 1 + |a|
  \end{equation*}
  for any $a \in \D$. By using again the reverse triangle inequality, this yields 
  \begin{align*}
    |\mu_j| &= \frac{1}{\sqrt{1-|a|^2}}\left| (1 + \overline{a}\omega^{j}) \dlag{u}[0] + \sum_{k=1}^{N-1} \left(\omega^{jk} + \overline{a} \omega^{j ( k+1 \ \textrm{mod} \ N)}\right) \dlag{u}[k] \right| \\ &\geq \frac{1 - |a|}{\sqrt{1 - |a|^2}} \left| |\dlag{u}[0]| - \discs\right|,
  \end{align*}
  where $\discs$ is defined according to \cref{eq:wellCondCond}. 
  If the diagonal dominance condition from \cref{eq:wellCondCond} holds, then we have
  \begin{equation*}
    |\mu_j| \geq \frac{1 - |a|}{\sqrt{1 - |a|^2}} (|\dlag{u}[0]| - \discs) \quad (j=0,\ldots,N-1).
  \end{equation*}
  By analogous arguments, we obtain
  \begin{equation*}
    |\mu_j| \leq \frac{1 + |a|}{\sqrt{1 - |a|^2}} (|\dlag{u}[0]| + \discs) \quad (j=0,\ldots,N-1).
  \end{equation*}
  Exploiting that $\Gamma$ is normal, we conclude
  \begin{equation*}
    \kappa_2(\Gamma) = \frac{\max_{j=0,\ldots,N-1}|\mu_j|}{\min_{j=0,\ldots,N-1} |\mu_j|} \leq \frac{1 + |a|}{1 - |a|}\cdot\frac{|\dlag{u}[0]| + \discs}{|\dlag{u}[0]| - \discs}.
  \end{equation*}
\end{proof}

\section{Implementation Details}
\label[appendix]{app:impl}

\subsection{General Remarks on Experiment Design}

In both of the considered experimental systems, the parameters $\{ \lambda_k \}_{k=0}^{\ord-1} \subset \D$ and the
corresponding residues $\{\res_k\}_{k=0}^{\ord-1} \subset \C$ are chosen randomly according to a uniform distribution.
The parameters $\{ \lambda_k \}_{k=0}^{\ord-1} \subset \D$ and $\{\res_k\}_{k=0}^{\ord-1} \subset \C$ define a causal
\ac{bibo}-stable \ac{siso} \ac{lti} system whose transfer function can be expressed according to \cref{eq:transferfunc}.
We note that for our experiments, the parameters $\lambda_k$ and $\res_k$ are chosen as complex conjugate pairs for
$k=0,\ldots,\ord-1$ to ensure that the system's transfer function is real-rational. All of our experiments are evaluated
on an Apple MacBook Pro (M4, 32 GB RAM, macOS Tahoe) using MathWorks MATLAB 2025a. The code of the proposed experiments is
available, and our results are fully reproducible; see the code statement at the end of the manuscript.

\subsection{Implementation Details for Experiments with the Large System}

Next, we detail some implementation choices for the experiments whose results are detailed in~\cref{tab:exp}. In
particular, we show how to construct the excitation that fulfils the conditions in~\cref{eq:expInputProps}. Consider an
$N$-point equidistant sampling of the interval $[0, 10]$, i.e., let $t_0 \coloneqq 0$ and
$t_n \coloneqq t_0 + n\cdot \Delta t$, where $\Delta t \coloneqq 10/N$. First, define the sequences $s = {(s_n)}_{n \ge 0}, w = {(w_n)}_{n \ge 0}\in \ell_2$ as
\begin{equation*}
s_n \coloneqq \begin{cases}
  \eu^{-t_n^2 / 2}, & 0 \leq n \leq N, \\
  0, & \textrm{otherwise},
\end{cases}
\end{equation*}
and
\begin{equation*}
  w_n \coloneqq \begin{cases}
    \eu^{-t_n^2 / 6}, & 0 \leq n \leq N, \\
    0, & \textrm{otherwise}.
  \end{cases}
\end{equation*}
Define furthermore the constant
\begin{equation*}
  \alpha \coloneqq -\frac{\cZ[w](1)}{\cZ[s](1)}.
\end{equation*}
Then, the input sequence can be given as
\begin{equation}
  \label{eq:slowdecInp}
  u_n \coloneqq \begin{cases}
    \alpha \cdot s_n + w_n, & 0 \leq n \leq N, \\
    0, & \textrm{otherwise}.
  \end{cases}
\end{equation}
This input is clearly in $\ell_2$ and its $\cZ$-transform satisfies
\begin{equation}
  \label{eq:Ugone}
  U(1) = \cZ[u](1) = \alpha \cdot \cZ[s](1) + \cZ[w](1) = -\frac{\cZ[w](1)}{\cZ[s](1)} \cdot \cZ[s](1) +  \cZ[w](1) = 0.
\end{equation}
The considered input sequence $u$ is illustrated in \cref{fig:expu-appendix}.
\begin{figure}[bth]
  \begin{subfigure}[t]{0.5\linewidth}
    \centering
    \parbox[t][130pt][t]{\linewidth}{\vspace{0pt}\centering\includegraphics{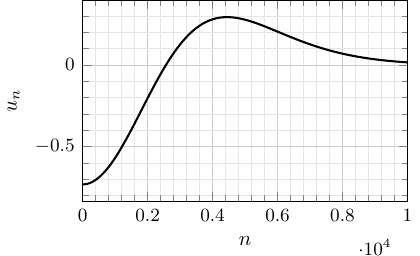}}
  \subcaption{Time domain}
  \end{subfigure}
  \begin{subfigure}[t]{0.5\linewidth}
    \centering
    \parbox[t][130pt][t]{\linewidth}{\vspace{0pt}\centering\includegraphics{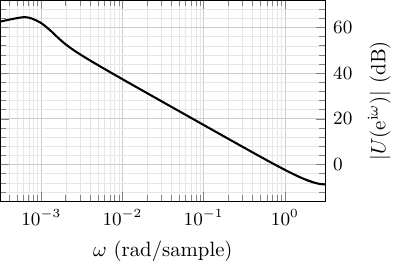}}
    \subcaption{Frequency domain}
  \end{subfigure}
  \caption{Experiment~2 input signal generated according to~\cref{eq:slowdecInp}. By construction, $U(1)=0$, hence \ac{etfe} cannot be applied directly.}
  \label{fig:expu-appendix}
\end{figure}
\begin{figure}[bth]
  \begin{subfigure}[t]{0.5\linewidth}
    \centering
    \parbox[t][130pt][t]{\linewidth}{\vspace{0pt}\centering\includegraphics{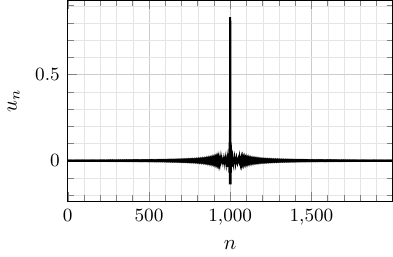}}
    \subcaption{Time domain.}
  \end{subfigure}
  \begin{subfigure}[t]{0.5\linewidth}
    \centering
    \parbox[t][130pt][t]{\linewidth}{\vspace{0pt}\centering\includegraphics{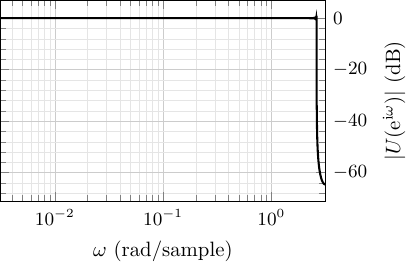}}
    \subcaption{Frequency domain.}
  \end{subfigure}
  \caption{Experiment~3 band-limited input signal~\cref{eq:sinc}.}
  \label{fig:sincInp}
\end{figure}
%

%% file: main.bbl
\begin{thebibliography}{10}

\bibitem{cdcAngino}
A.~Angino, T.~Dózsa, and M.~Voigt.
\newblock {{${H}_2$}}-optimal model order reduction using hyperbolic geometry.
\newblock In {\em 2026 {IEEE} 65th {Conference} on {Decision} and {Control}
  ({CDC})}, Honolulu, HI, USA, 2026.
\newblock Accepted for presentation.

\bibitem{aronszajn1950}
N.~Aronszajn.
\newblock Theory of reproducing kernels.
\newblock {\em Transactions of the American Mathematical Society},
  68(3):337--404, 1950.
\newblock doi: \url{10.1090/S0002-9947-1950-0051437-7}.

\bibitem{Candan2011}
C.~Candan.
\newblock On the eigenstructure of {DFT} matrices [{DSP} {Education}].
\newblock {\em IEEE Signal Processing Magazine}, 28(2):105--108, 2011.
\newblock doi: \url{10.1109/MSP.2010.940004}.

\bibitem{dozsa2024}
T.~Dózsa, F.~Schipp, and A.~Soumelidis.
\newblock On {Bernoulli}’s method.
\newblock {\em SIAM Journal on Numerical Analysis}, 62(3):1259--1277, 2024.
\newblock doi: \url{10.1137/22M1528501}.

\bibitem{dzrbasjan1962}
M.~M. Džrbašjan.
\newblock Expansions in rational functions with fixed poles.
\newblock {\em Doklady Akademii Nauk SSSR}, 143(1):17--20, 1962.

\bibitem{farina2000}
A.~Farina.
\newblock Simultaneous measurement of impulse response and distortion with
  swept-sine technique.
\newblock In {\em 108th {AES} {Convention}, 5093}, Paris, France, 2000.

\bibitem{fridli2020}
S.~Fridli and F.~Schipp.
\newblock Discrete rational biorthogonal systems on the disc.
\newblock {\em Annales Universitatis Scientiarum Budapestinensis de Rolando
  Eötvös Nominatae. Sectio Computatorica}, 50:127--134, 2020.
\newblock doi: \url{10.71352/ac.50.127}.

\bibitem{fujimoto2020}
Y.~Fujimoto.
\newblock Time-frequency regularization for impulse response estimation.
\newblock In {\em 2020 59th {Annual} {Conference} of the {Society} of
  {Instrument} and {Control} {Engineers} of {Japan} ({SICE})}, pages
  1329--1332, Chiang Mai, Thailand, 2020.
\newblock doi: \url{10.23919/SICE48898.2020.9240275}.

\bibitem{Garcia2018}
S.~R. Garcia, J.~Mashreghi, and W.~T. Ross.
\newblock {\em Finite {Blaschke} {Products} and {Their} {Connections}}.
\newblock Springer, Cham, Switzerland, 2018.
\newblock doi: \url{10.1007/978-3-319-78247-8}.

\bibitem{golub2013}
G.~H. Golub and C.~F. Van~Loan.
\newblock {\em Matrix {Computations}}.
\newblock Johns {Hopkins} {Studies} in the {Mathematical} {Sciences}. The Johns
  Hopkins University Press, Baltimore, MD, USA, 4th edition, 2013.

\bibitem{guarnizo2018}
C.~Guarnizo and M.~A. Álvarez.
\newblock Impulse response estimation of linear time-invariant systems using
  convolved {Gaussian} processes and {Laguerre} functions.
\newblock In M.~Mendoza and S.~Velastín, editors, {\em Progress in {Pattern}
  {Recognition}, {Image} {Analysis}, {Computer} {Vision}, and {Applications}},
  volume 10657 of {\em Lecture {Notes} in {Computer} {Science}}, pages
  281--288. Springer, Cham, Switzerland, 2018.
\newblock doi: \url{10.1007/978-3-319-75193-1_34}.

\bibitem{gustavsen1999}
B.~Gustavsen and A.~Semlyen.
\newblock Rational approximation of frequency domain responses by vector
  fitting.
\newblock {\em IEEE Transactions on Power Delivery}, 14(3):1052--1061, 1999.
\newblock doi: \url{10.1109/61.772353}.

\bibitem{hansen2002}
P.~C. Hansen.
\newblock Deconvolution and regularization with {Toeplitz} matrices.
\newblock {\em Numerical Algorithms}, 29(4):323--378, 2002.
\newblock doi: \url{10.1023/A:1015222829062}.

\bibitem{he2020}
N.~He and C.~Shen.
\newblock Frequency constrained predictive control for large scaled spatially
  distributed systems.
\newblock {\em Control Engineering Practice}, 100:104440, 2020.
\newblock doi: \url{10.1016/j.conengprac.2020.104440}.

\bibitem{illg2024}
C.~Illg and O.~Nelles.
\newblock Regularized finite impulse response models versus {Laguerre} models:
  {A} comparison.
\newblock {\em IFAC-PapersOnLine}, 58(15):67--72, 2024.
\newblock doi: \url{10.1016/j.ifacol.2024.08.506}.

\bibitem{kailath1980}
T.~Kailath.
\newblock {\em Linear {Systems}}.
\newblock Prentice-{Hall} {Information} and {System} {Science} {Series}.
  Prentice-Hall, Englewood Cliffs, NJ, USA, 1980.

\bibitem{kilmer1999}
M.~E. Kilmer and D.~P. O'Leary.
\newblock Pivoted {Cauchy}-like preconditioners for regularized solution of
  ill-posed problems.
\newblock {\em SIAM Journal on Scientific Computing}, 21(1):88--110, 1999.
\newblock doi: \url{10.1137/S1064827596308974}.

\bibitem{ljung1985}
L.~Ljung.
\newblock On the estimation of transfer functions.
\newblock {\em Automatica}, 21(6):677--696, 1985.
\newblock doi: \url{10.1016/0005-1098(85)90042-1}.

\bibitem{ljung1999}
L.~Ljung.
\newblock {\em System {Identification}: {Theory} for the {User}}.
\newblock Prentice {Hall} {Information} and {System} {Sciences} {Series}.
  Prentice Hall, Upper Saddle River, NJ, USA, 2nd edition, 1999.

\bibitem{lorentz1996}
G.~G. Lorentz, M.~v. Golitschek, and Y.~Makovoz.
\newblock {\em Constructive {Approximation}: {Advanced} {Problems}}, volume 304
  of {\em Grundlehren der mathematischen {Wissenschaften}}.
\newblock Springer, Berlin, Heidelberg, Germany, 1996.

\bibitem{marconato2017}
A.~Marconato, M.~Schoukens, and J.~Schoukens.
\newblock Filter-based regularisation for impulse response modelling.
\newblock {\em IET Control Theory \& Applications}, 11(2):194--204, 2017.
\newblock doi: \url{10.1049/iet-cta.2016.0908}.

\bibitem{muller-trapet2020}
M.~Müller-Trapet.
\newblock On the practical application of the impulse response measurement
  method with swept-sine signals in building acoustics.
\newblock {\em The Journal of the Acoustical Society of America},
  148(4):1864--1878, 2020.
\newblock doi: \url{10.1121/10.0001916}.

\bibitem{pillonetto2014}
G.~Pillonetto, F.~Dinuzzo, T.~Chen, G.~De~Nicolao, and L.~Ljung.
\newblock Kernel methods in system identification, machine learning and
  function estimation: {A} survey.
\newblock {\em Automatica}, 50(3):657--682, 2014.
\newblock doi: \url{10.1016/j.automatica.2014.01.001}.

\bibitem{qian2011}
T.~Qian and Y.-B. Wang.
\newblock Adaptive {Fourier} series—a variation of greedy algorithm.
\newblock {\em Advances in Computational Mathematics}, 34(3):279--293, 2011.
\newblock doi: \url{10.1007/s10444-010-9153-4}.

\bibitem{rebillat2018}
M.~Rebillat and M.~Schoukens.
\newblock Comparison of least squares and exponential sine sweep methods for
  {Parallel} {Hammerstein} {Models} estimation.
\newblock {\em Mechanical Systems and Signal Processing}, 104:851--865, 2018.
\newblock doi: \url{10.1016/j.ymssp.2017.11.015}.

\bibitem{soumelidis2017}
A.~Soumelidis, J.~Bokor, F.~Schipp, and Z.~Szabó.
\newblock Hyperbolic geometrical approach to model reduction.
\newblock {\em IFAC-PapersOnLine}, 50(1):12905--12910, 2017.
\newblock doi: \url{10.1016/j.ifacol.2017.08.1784}.

\bibitem{soumelidis1991}
A.~Soumelidis and A.~Edelmayer.
\newblock Modelling of complex systems for control and fault diagnostics: {A}
  knowledge based approach.
\newblock In S.~G. Tzafestas, editor, {\em Engineering {Systems} with
  {Intelligence}}, volume~9 of {\em Microprocessor-{Based} and {Intelligent}
  {Systems} {Engineering}}, pages 125--132. Springer, Dordrecht, The
  Netherlands, 1991.
\newblock doi: \url{10.1007/978-94-011-2560-4_14}.

\bibitem{trefethen2014}
L.~N. Trefethen and J.~A.~C. Weideman.
\newblock The exponentially convergent trapezoidal rule.
\newblock {\em SIAM Review}, 56(3):385--458, 2014.
\newblock doi: \url{10.1137/130932132}.

\bibitem{tuma2019}
M.~Tuma and P.~Jura.
\newblock Impulse response approximation of dead time {LTI} {SISO} systems
  using generalized {Laguerre} functions.
\newblock {\em AIP Conference Proceedings}, 2116(1):310010, 2019.
\newblock doi: \url{10.1063/1.5114317}.

\bibitem{vandenhof2005}
P.~Van~den Hof and B.~Ninness.
\newblock System {Identification} with {Generalized} {Orthonormal} {Basis}
  {Functions}.
\newblock In P.~S. Heuberger, P.~M. Van~den Hof, and B.~Wahlberg, editors, {\em
  Modelling and {Identification} with {Generalized} {Orthonormal} {Basis}
  {Functions}}. Springer, London, UK, 2005.
\newblock doi: \url{10.1007/1-84628-178-4_4}.

\bibitem{vanoverschee1994}
P.~Van~Overschee and B.~De~Moor.
\newblock {N4SID}: {Subspace} algorithms for the identification of combined
  deterministic-stochastic systems.
\newblock {\em Automatica}, 30(1):75--93, 1994.
\newblock doi: \url{10.1016/0005-1098(94)90230-5}.

\bibitem{Venkataraman2019}
R.~Venkataraman and P.~Seiler.
\newblock System identification for a small, rudderless, fixed-wing unmanned
  aircraft.
\newblock {\em Journal of Aircraft}, 56(3):1126--1134, 2019.
\newblock doi: \url{10.2514/1.C035141}.

\bibitem{verhaegen1992}
M.~Verhaegen and P.~Dewilde.
\newblock Subspace model identification {Part} 1. {The} output-error
  state-space model identification class of algorithms.
\newblock {\em International Journal of Control}, 56(5):1187--1210, 1992.
\newblock doi: \url{10.1080/00207179208934363}.

\end{thebibliography}
